\documentclass[sigconf,screen]{acmart}
\AtBeginDocument{%
  }

\fontsize{11}{213}\selectfont
\newcommand{\memory}{\mu}

\acmYear{2025}\copyrightyear{2025}
\acmConference[SPAA '25]{37th ACM Symposium on Parallelism in Algorithms and Architectures}{July 28--August 1, 2025}{Portland, OR, USA}
\acmBooktitle{37th ACM Symposium on Parallelism in Algorithms and Architectures (SPAA '25), July 28--August 1, 2025, Portland, OR, USA}
\acmDOI{10.1145/3694906.3743302}
\acmISBN{979-8-4007-1258-6/25/07}

\usepackage[linesnumbered, ruled,vlined]{algorithm2e}

\usepackage{colortbl}
\definecolor{bgcolor}{rgb}{0.65, 0.85, 0.99}

\newcommand{\E}{\mathbb{E}}

\newcommand{\set}[1]{\left\{#1\right\}}
\newcommand{\brackets}[1]{\left[#1\right]}

\def\<{\left\langle}
\def\>{\right\rangle}
\def\[{\left[}
\def\]{\right]}
\def\({\left(}
\def\){\right)}

\newcommand{\ie}{\textit{i.e.}}

\newcommand{\vol}{\operatorname{vol}}
\newcommand{\poly}{\operatorname{poly}}
\newcommand{\tmix}{\tau_{\operatorname{mix}}}

\newcommand{\congest}{\ensuremath{\mathsf{CONGEST}}\xspace}
\newcommand{\local}{\ensuremath{\mathsf{LOCAL}}\xspace}
\newcommand{\clique}{\ensuremath{\mathsf{Congested\ Clique}}\xspace}
\newcommand{\switch}{\ensuremath{\mu}-\congest}

\newcommand{\floor}[1]{\lfloor#1\rfloor}
\newcommand{\parentheses}[1]{\left(#1\right)}

\usepackage{thmtools}
\usepackage{thm-restate}

\usepackage{hyperref}
\usepackage{amsthm}
\usepackage{cleveref}

\begin{document}
\title{Bounded Memory in Distributed Networks}

\ccsdesc[500]{Theory of computation~Distributed algorithms}

\keywords{Distributed graph algorithms, Bounded memory, Subgraph detection, Streaming algorithms}

\author{Ran Ben Basat}
\orcid{0000-0003-0196-9190}
\affiliation{%
  \institution{University College London}
  \city{London}
  \country{UK}}
\email{r.benbasat@cs.ucl.ac.uk}

\author{Keren Censor-Hillel}
\orcid{0000-0003-4395-5205}
\affiliation{%
  \institution{Technion}
  \city{Haifa}
  \country{Israel}}
\email{ckeren@cs.technion.ac.il}

\author{Yi-Jun Chang}
\orcid{0000-0002-0109-2432}
\affiliation{%
  \institution{National University of Singapore}
  \country{Singapore}}
\email{cyijun@nus.edu.sg}

\author{Wenchen Han}
\orcid{0009-0001-8262-5658}
\affiliation{%
  \institution{University College London}
  \city{London}
  \country{UK}}
\email{wenchen.han.22@ucl.ac.uk}

\author{Dean Leitersdorf}
\orcid{0000-0002-2775-9207}
\affiliation{%
  \institution{Technion}
  \city{Haifa}
  \country{Israel}}
\email{dean.leitersdorf@gmail.com}

\author{Gregory Schwartzman}
\orcid{0000-0002-8461-1479}
\affiliation{%
  \institution{JAIST}
  \city{Nomi}
  \country{Japan}}
\email{gregory.schwartzman@gmail.com}

\newif\ifcomm
\ifcomm
\else
\commfalse
\fi
\ifcomm
    \newcommand\Wenchen[1]{\textcolor{blue}{Wenchen: #1}}
    \newcommand\ran[1]{\textcolor{red}{Ran: #1}}
    \newcommand\gregory[1]{\textcolor{purple}{Gregory: #1}}
    \newcommand\keren[1]{\textcolor{violet}{Keren: #1}}
    \newcommand\yijun[1]{\textcolor{cyan}{Yi-Jun: #1}}
    \newcommand\dtodo[1]{\textcolor{blue}{[Dean: #1]}}
    \newcommand{\CRdel}[1]{\textcolor{red}{\sout{#1}}}
    \definecolor{wenchen}{RGB}{0,0,200}
    \newcommand{\update}{\textcolor{wenchen}}
\else
    \newcommand{\mycomm}[3]{}
    \newcommand{\CRdel}[1]{}
    \newcommand\Wenchen[1]{}
    \newcommand\ran[1]{}
    \newcommand\gregory[1]{}
    \newcommand\keren[1]{}
    \newcommand\yijun[1]{}
    \newcommand\dtodo[1]{}
    \newcommand{\update}{}
\fi


\begin{abstract}
The recent advent of programmable switches makes distributed algorithms readily deployable in real-world datacenter networks. However, there are still gaps between theory and practice that prevent the smooth adaptation of \congest algorithms to these environments.
In this paper, we focus on the memory restrictions that arise in real-world deployments. We introduce the $\mu$-\congest model where on top of the bandwidth restriction, the memory of nodes is also limited to $\mu$ words, in line with real-world systems. We provide fast algorithms of two main flavors.

First, we observe that many algorithms in the \congest model are memory-intensive and do not work in $\mu$-\congest. A prime example of a family of algorithms that use large memory is clique-listing algorithms. We show that the memory issue that arises here cannot be resolved without incurring a cost in the round complexity, by establishing a lower bound on the round complexity of listing cliques in $\mu$-\congest. We introduce novel techniques to overcome these issues and generalize the algorithms to work within a given memory bound. Combined with our lower bound, these provide tight tradeoffs between the running time and memory of nodes. 

Second, we show that it is possible to efficiently simulate various families of streaming algorithms in $\mu$-\congest. These include fast simulations of $p$-pass algorithms, random order streams, and various types of mergeable streaming algorithms.

Combining our contributions, we show that we can use streaming algorithms to efficiently generate statistics regarding combinatorial structures in the network. An example of an end result of this type is that we can efficiently identify and provide the per-color frequencies of the frequent monochromatic triangles in $\mu$-\congest.

\end{abstract}
\maketitle
\ifcomm
\fi



\ifcomm
\fi





\section{Introduction}
\label{sec:intro}

Traditionally, network switches were fixed-purpose hardware that typically require manufacturing new chips for any customization. While some distributed algorithms (e.g., Spanning Tree Protocol \cite{IEEESTP}) were used in the past, there was a significant industry shift to Software Defined Networking (SDN) in which almost all network functionalities are performed on a centralized controller that optimizes the configuration and updates the switches~\cite{hong2018b4, ferguson2021orion, mckeown2008openflow, medved2014opendaylight, berde2014onos}.

Quite recently, a new generation of switches emerged that allows programmability \cite{bosshart2013forwarding, Tofino, bosshart2014p4,trident}. These switches allow us to compile code that modifies their behavior without changing the hardware.
This opens opportunities to leverage advances in theoretical distributed algorithms, allowing quick prototyping and evaluation of new ideas. 
Essentially, by running distributed algorithms, one can avoid going through the controller and optimize the reaction time to network changes (e.g., link/switch failures). 
Indeed, recently, \cite{HanFSMYB22} showed the benefit of adapting some theoretical algorithms to these switches in several networking applications, such as clock synchronization and source-routed multicast, in which reaction time is paramount.
For example, they show that distributed algorithms such as graph spanners can be leveraged to react to network changes two orders of magnitude \mbox{faster than traditional SDN-based solutions.}



However, the crux is that implementing algorithms on these switches introduces some challenges that are not accounted for in most theory models, the main ones being:
(i) these switches have a limit on the number of operations they can invoke per arriving packet, i.e., they have a limit on their computational abilities, and
(ii) their memory is limited.

In this paper, we take a first step towards modeling the limited resources available at these switches. Our intent is to allow theory researchers to design fast algorithms that will be easier to implement on switches. 
Current datacenter networks can have 10s-100s of thousands of switches~\cite{guo2015pingmesh}, each with 10s-100s of ports (neighbors in the communication graph)~\cite{Tofino,trident}. The local memory for each switch is often measured in just 10s of MB~\cite{Tofino}, as we must use a very fast SRAM memory to allow accessing it while processing packets. Thus, it is reasonable to work with roughly quasi-linear memory. Yet, to allow our abstraction to have wider applicability, we will parametrize the memory bound by $\mu$, as follows.

\paragraph{The $\mu$-\congest model.} Our abstraction is based on the classic \congest model~\cite{PelegBook}, in which the network is represented by an undirected graph $G = (V, E)$, where $n$ nodes are computational units (switches) and edges are communication links. Communication proceeds in synchronous rounds, in which each node is allowed to send an $O(\log n)$-bit message to each of its neighbors. The main complexity measure is the number of rounds required for an algorithm to complete. The point in which the $\mu$-\congest model deviates from the classic \congest model is that each node only has $\mu$ words (each of $O(\log n)$ bits) of memory to store the inputs and some auxiliary variables. For outputs, once an output word of a \switch algorithm is ready, it will be emitted and no longer take any of the nodes' memory. This is well motivated both from a theoretical and a practical point of view. For example, consider the fundamental problem of triangle listing: a node might need to output a number of triangles that is much larger than its memory. From a practical perspective, it is common that intermediate inputs are transmitted to a higher-level application. We note that while many \congest algorithms are naturally memory efficient as they require a small number of rounds (e.g., nodes can store the received messages), some applications require a substantial amount of memory, as we exemplify with $k$-clique listing below.


%
%

Throughout the paper, unless otherwise stated, we assume that $\mu$ is at least the maximum degree $\Delta$. If we allow the memory size $\mu$ of a node $v$ to be smaller than $\deg(v)$, then the complexity of a problem might depend on the model of how the incoming messages arrive within a round. For example, in the \emph{adversarial} setting,  the $\deg(v)$ incoming messages may arrive in any arbitrary order and $v$ needs to process these messages in a streaming fashion under the memory constraint $\mu$.

When $\mu=\omega(1)$, we may implicitly use $O(\mu)$ memory instead of $\mu$, for simplicity of proofs. In such cases, it is always possible to use $\mu$ memory by running the proof with $\mu^{'} = \mu/c$, for some constant $c$.  

\paragraph{Related work.}
To the best of our knowledge, prior work on distributed algorithms in general networks with low memory focused specifically on routing tables, in a rich line of work (see, e.g., ~\cite{PelegU89, AwerbuchBLP90, AwerbuchP92, FraigniaudG95, FraigniaudG96, GavoilleP96, FraigniaudG06, LenzenP13, LenzenP15, ElkinN18, ElkinN18b, CastanedaDT18,CastanedaLT20}).

A distributed model of computation that is tightly-related to \switch is the Massively Parallel Computation model (MPC), introduced in~\cite{KarloffSV10} and widely studied for abundant computational tasks. In MPC, the computing machines are also restricted in terms of their local memory. However, the communication network is an all-to-all network, as opposed to an arbitrary graph that is allowed in \switch. Not being able to rely on having a complete communication graph substantially increases the difficulties in the computational challenges that arise in the \switch model compared to the MPC model, and better corresponds to the real-world settings that we wish to abstract in this work. We do address an all-to-all model in Section~\ref{subsec:listingUBcc}, but this is done as a warm-up for our \switch algorithms rather than as an end goal.

Our model is also related to the Stone Age Distributed Computing model by Emek and Wattenhofer~\cite{EmekW13} where communication is asynchronous, and nodes are implemented as networked finite state machines (nFSM) and thus work with constant-size memory.

\subsection{Our Contributions}
\label{subsec:contributions}
While some algorithms in the classic \congest model already use low memory \update{to be able to fit into the limited $O(\mu)$ memory constraint}, in others, there is at least one node that collects a large amount of information. Naturally, our goal in this work is to investigate the influence of memory limitations on the latter. \update{Intuitively, as a tradeoff for the reduced memory available for the distributed algorithm, the lower bound of the round complexity required (that reflects the amount of computation being performed on the nodes) will increase accordingly. } \Wenchen{Add some high-level insights regarding the tradeoff between the decreased $\mu$ and increased round complexity / compute, and add some more introductions.}

\update{In the following sections, we showcase how we establish lower bounds and upper bounds (as algorithmic solutions) of distributed graph problems under the \switch model to provide more insights on how limited memory affects the round complexity of the \switch algorithms.}

\subsubsection{Clique listing} \quad


\paragraph{Background.} Finding $k$-cliques\footnote{For simplicity, we assume throughout the paper that $k$ is a constant.} is a prominent task in many computational settings. In the well-studied variant of \emph{listing} or \emph{enumerating} $k$-cliques, each node is required to output a set of $k$-cliques, such that the union of all outputs equals the set of all $k$-cliques in the input graph. In fact, listing $k$-cliques is an immediate example of a task whose algorithms use a rather large local memory in the classic \congest model. For example, optimal triangle listing completes in $\widetilde{O}(n^{1/3})$ rounds, and uses $\widetilde\Theta(n^{4/3})$ memory per node \cite{ChangPSZ21,DolevLP12}. It remains unclear, though, whether such memory usage is optimal (what the lower bound is), and how we could design a listing algorithm if the memory size is smaller than $\widetilde\Theta(n^{4/3})$.

\paragraph{Our results.} 
We give a lower bound on the round complexity of algorithms that use $\mu=o(n^{4/3})$ memory.

Notice that our lower bound requires a value $\ell$ which is an upper bound on the number of messages each node receives per round -- in \switch, this, of course, is at most $\Delta$, and in another all-to-all version of \switch which we introduce later, this is at most $n$.

\begin{restatable}{theorem}{ListingLB}
\label{thm: lb-listing}
Let $\ell \leq n$ be an upper bound on the number of messages any node can receive per round. Given $\ell \leq \memory\leq n^{2-2/k}$,
the round complexity for any $k$-clique listing algorithm in $\mu$-\congest is at least $\Omega(n^{k-1}/(\memory^{k/2-1}\cdot \ell)) = \Omega(n^{k-2}/\memory^{k/2-1})$. 
\end{restatable}

For example, for triangle listing ($k=3$), we get a bound of $\Omega(n / \mu^{1/2})$.
Indeed, we then show that our lower bound faithfully characterizes the trade-off between $\mu$ and the round complexity of triangle listing, by providing a randomized algorithm that completes within a corresponding round complexity, up to subpolynomial factors, w.h.p.\footnote{With high probability stands for a probability that is at least $1-1/n^c$ for some constant $c\geq 1$.} While our algorithms follow the high-level framework of their \congest counterparts, several intriguing technical challenges arise, which we overcome using new methods.

\begin{restatable}{theorem}{ListingUBcongest}
\label{thm: ub-listing-congest}
Given $\memory \le n^{4/3}$, there exists a 
algorithm for listing all $3$-cliques (triangles) within $n^{1+o(1)}/\memory^{1/2}$ rounds, w.h.p.
\end{restatable}

Exemplifying our results for a linear memory of $\mu=\Theta(n)$ gives that the complexity of triangle listing in this case is $\Theta(n^{1/2+o(1)})$.


\subsubsection{Simulating Streaming Algorithms on a Single Node} \quad

\paragraph{Background.} A successful approach for handling memory limitations in centralized computations is to use streaming algorithms. This is a setting in which the computational unit receives the input as a stream of bits and processes it along the way, under the restriction that it can only store $M$ words of $O(\log n)$ bits in memory at any given time. In the following context we further assume that the input data only consists of edges (and their respective properties) of the graph, and we call this the edge-streaming model. 

We study how to \textit{simulate} edge-streaming algorithms in the distributed setting on a graph, using the \switch model. A single-node \switch simulation is defined as an algorithm in which only a single node is responsible for the execution of the streaming algorithm, \ie, it receives the inputs of all nodes as an input stream, executes the \mbox{streaming algorithm, and outputs the result.} 


\paragraph{Our results.} To this end, a na\"ive approach would be to \textit{directly} send all graph information (edges, weights, etc.) to the simulator node. This is clearly the best that can be done in general, yet if the streaming algorithm makes multiple passes over the stream, say $p\geq 1$ times, then repeating this procedure incurs an expensive $O(n\cdot\Delta\cdot p)$ complexity. We show that one can do better when simulating $p$-pass streaming algorithms, by observing that the neighbors of the simulator node collectively have sufficient memory for storing the required information. {We leverage this for speeding up subsequent passes, to obtain the following results.}

\begin{restatable}{theorem}{StreamingSimulation}
\label{thm:streamingSimulation}
For any $p$-pass edge-streaming algorithm that uses $M$ memory, there exists a single-node \switch simulation within $O(n \cdot (\Delta + p))$ rounds using $\mu =M+n$ memory.
\end{restatable}

We further show that unless the memory $\memory$ is super-linear, one cannot do better than our above approach, as follows.
\begin{restatable}{theorem}{StreamingSimulationLB}
\label{thm:streaminSimulationLB}
For $\mu \le n/4$, any \switch algorithm that simulates a $p$-pass edge-streaming algorithm on one node requires $\Omega(n\cdot \Delta \cdot p)$ rounds.
\end{restatable}


\paragraph{Applications.} Our simulation result is applicable for a wide range of semi-streaming algorithms, where any $p=\omega(1)$-pass algorithm benefits from our ``edge caching'' method.
Here, we give a couple of examples.
The work of Liu et al.~\cite{2020arXiv200906106L} shows that maximum weight matching can be solved in $\widetilde O(\sqrt m)$ passes and $\widetilde O(n)$ space. By simulating this algorithm, we get that it admits a $\widetilde O(n\cdot (\Delta+\sqrt m))$-round \switch algorithm for $\mu=n\ \mbox{polylog}\ n$.
%

~\\It is well-known that in many cases, the order of the stream elements directly affects the guarantees of streaming algorithms. 
In particular, random-order streams play an important role in this aspect. However, in some cases it could be less reasonable to assume that the order of the stream is random. Interestingly, in the \switch model, we show that we can \emph{generate} a random order stream for the simulator node, within the same round complexity as our \mbox{general simulation result, as stated next.}

\begin{restatable}{theorem}{StreamingRandom}
\label{thm:random-order}
    For any $p$-pass \emph{random-order} edge-streaming algorithm that uses $M$ memory, there exists a 
    method that simulates it within $O(n \cdot (\Delta + p))$ rounds using $\mu = M+n+\Delta^2$ memory.
\end{restatable}

\subsubsection{Mergeable Streaming Simulations}\quad

\paragraph{Background.} Finally, we consider the case where the simulated streaming algorithm adheres to certain \emph{mergeability}~\cite{agarwal2013mergeable} properties. With such properties, we could leverage the parallel nature of the \switch model in order to execute a simulation of streaming algorithms throughout the entire network, rather than only on a single simulator node. This way, we can further reduce the round complexity. In what follows, we start by having $t_v$ words of information in each node $v$, and the goal is to compute some summary of all the information in the graph. 

The above is achieved by computing intermediate summaries of the data and then \emph{merging} them. We consider three notions of mergeability~\cite{agarwal2013mergeable}, and provide an intuitive explanation for each below (defined formally in Section~\ref{sec: mergeable streaming}).
\begin{enumerate}
    \item One way mergeability - It is possible to merge all summaries into one main summary (e.g., $S_1$ can be merged with $S_2$ to get $S'$ and then $S'$ can be merged with $S_3$ and so on).
    \item Full mergeability - the summaries support arbitrary merging patterns (e.g., $S_1$ can be merged with $S_2$ and $S_3$ with $S_4$ before merging the two results).
    \item Composable - it is possible to merge multiple $M$-sized summaries into a single $M$-sized summary in a streaming fashion without requiring more than $M$ space.
\end{enumerate}

\paragraph{Our results.} We state the following theorems and applications.

\begin{restatable}{theorem}{oneway}
\label{thm:oneway}
    The deterministic round complexity of simulating a one-way mergeable streaming algorithm in \switch is $O\parentheses{\min \set{n\cdot M,\sqrt{|\mathcal{I}|\cdot M}} + D}$, where $M$ is the size of the summary, $|\mathcal{I}|=\sum_{v\in V} t_v$, $D$ is the diameter of the graph, and $\mu = \Omega(\Delta + M)$.
\end{restatable}

\paragraph{Applications.}
Consider having integer labels of length $O(\log n)$ on each edge of the graph and deterministically computing approximate quantiles of the set of labels. By using a quantile sketch such as the GK algorithm~\cite{greenwald2001space}, and leveraging its one-way mergeability~\cite{agarwal2013mergeable}, we can estimate quantiles up to an $m\cdot \epsilon$ (which gives $M=O(\epsilon^{-1}\cdot\log(m\cdot \epsilon))$) additive error using $O\parentheses{\sqrt{m\cdot \epsilon^{-1}\cdot \log(m\epsilon)} + D}$ rounds and $\mu=O(\epsilon^{-1}\log(m\epsilon) + \Delta)$ space.


Next, we consider a stronger notion of mergeability, using which we may hierarchically merge summaries in parallel to accelerate the computation.
\begin{restatable}{theorem}{fullym}
\label{thm:fullym}
    The deterministic round complexity of simulating a fully-mergeable streaming algorithm in \switch is $O\parentheses{\parentheses{\log\min\set{n\cdot M, |\mathcal I|}}\cdot\parentheses{M\cdot \log\frac{\Delta}{\mu / M}+D}}$, where $M$ is the size of the summary, $|\mathcal{I}|=\sum_{v\in V} t_v$, $D$ is the diameter of the graph and $\mu = \Omega(\Delta + M)$.
\end{restatable}

\paragraph{Applications.}
Consider having integer labels of length $O(\log n)$ on each edge of the graph and deterministically computing the most frequent labels and their frequencies. By using a heavy hitters sketch algorithm such as the MG sketch~\cite{misra1982finding}, which is known to be fully-mergeable~\cite{agarwal2013mergeable}, we can estimate the frequency of each label to within an $m\cdot \epsilon$ additive error (which gives $M=\epsilon^{-1}$) using $O\parentheses{{\log m\cdot\parentheses{\epsilon^{-1}\cdot \log\parentheses{\frac{\Delta\epsilon^{-1}}{\mu }}+D}}}$ \mbox{rounds and $\mu=O(\Delta + \epsilon^{-1})$ space.}

We can also find all labels that appear at least $m\cdot \epsilon$ times and their \emph{exact} frequencies in the same round complexity. To that end, we first use the above algorithm to compute the estimated frequencies within an $\epsilon/3$ additive error. Next, we consider the set of labels whose estimated frequencies are at least $2\epsilon/3$. Notice that we have at most $3/\epsilon$ such labels and that the target set of labels is among these. Then, we count the frequencies of these labels exactly by propagating their count up a BFS tree. Note that this requires just $O(\epsilon^{-1}+D)$ rounds and $O(\Delta+\epsilon^{-1})$ space.

Finally, we consider the case where the algorithm is also composable, which allows a streaming aggregation of summaries to facilitate a further speedup.
\begin{restatable}{theorem}{composable}
\label{thm:composable}
    The deterministic round complexity of simulating a fully-mergeable and composable streaming algorithm in \switch is $O\parentheses{\parentheses{\log\min\set{n\cdot M, |\mathcal I|}}\cdot\parentheses{M+D}}$, where $M$ is the size of the summary, $|\mathcal{I}|=\sum_{v\in V} t_v$ and $\mu = \Omega(\Delta + M)$.
\end{restatable}

\paragraph{Applications.}
All linear sketches are composable in a straightforward manner. While most such solutions are randomized, there are also deterministic sketches such as CR-Precis~\cite{ganguly2007cr}. This sketch allows one to estimate a variety of metrics over the underlying data.
For example, for parameters $\epsilon\in(0,1),\alpha>1$, using the CR-Precis sketch we can produce an estimator $\widehat H$, of the entropy $H$ of the edge label distribution, that satisfies 
\begin{equation*}
    \frac{H(1-\epsilon)}{\alpha} \le \widehat H \le \parentheses{1+\frac{\epsilon}{\sqrt{\log m}}}\cdot H,
\end{equation*}
using a round complexity of $O\parentheses{\log m\cdot\parentheses{M+D}}$ and memory $\mu=O(M+\Delta)$, where
\begin{equation*}
M = \log^2 m \cdot \epsilon^{-4} \cdot m^{2/\alpha}\cdot \log(m\epsilon^{-1}).
\end{equation*}

\subsection{Technical Overview}
\label{subsec:overview}

\subsubsection{Clique listing}\quad

\paragraph{Lower bounds for clique listing (\cref{thm: lb-listing}).} Here we overview the challenge and approach for the case of $k=3$, but the proof applies to all constant values of $k$. Essentially, we would like to follow the outline of the lower bound of~\cite{IzumiLG17}, which very roughly speaking works as follows. It first bounds from below by $\widetilde\Omega(n^{4/3})$ bits the worst-case amount of information that a node needs to receive throughout a triangle-listing algorithm in a random graph in order to output sufficiently many triangles. This is because there are $\Theta(n^3)$ triangles and so some node needs to output $\Omega(n^2)$ triangles, for which it needs $\widetilde\Omega(n^{4/3})$ bits of information. The argument then concludes that since a node can only receive $\widetilde O (n)$ bits per round, the complexity of the algorithm must be $\widetilde\Omega(n^{1/3})$. 

When we bound the memory at each node, we need to adjust the above outline to analyze the number of triangles that can be output by a node per round (with $\memory$ memory), rather than throughout the entire algorithm. To this end, we would like to say that if an algorithm completes in $L$ rounds, then for the node that needs to output $\Omega(n^2)$ triangles there needs to be a round in which it outputs $\Omega(n^2/L)$ triangles, and so because of the amount of information it has to store for this, we obtain a lower bound on $L$ in terms of $n$ and $\mu$. This comes with a caveat: if we change our viewpoint to consider \emph{only} the communication overhead that arises from the bound on the local memory then we miss the cost of communication for obtaining the information at the node. A bit more precisely, if the memory is larger than the incoming bandwidth, i.e., if roughly $\mu\geq n$, then an approach that only looks at the output ``per round'' does not take into account the number of rounds required for \emph{collecting} $\memory$ memory in the first place. Specifically, such an approach loses a factor of $\memory/n$ in the lower bound it obtains.

We overcome this challenge by considering a slightly modified setup. In this new setup, the memory size is $3\memory$, but a node is only allowed to output triangles at the end of every (non-overlapping) window of $\memory/n$ rounds. We prove that (i) our claimed lower bound holds for this setting without losing the $\memory/n$ factor because we provide this window size for collecting $\memory$ information, and that (ii) any algorithm in the original setup, with $\memory$ memory and no restriction on when a node can output triangles, can be simulated in the new setup with no overhead. These yield our statement.


\paragraph{Expander decomposition and routing (elaborated in \Cref{sect:expander-appendix})} To show an optimal algorithm for triangle listing in \switch, we first adapt existing expander decomposition and routing algorithms from \congest to \switch. Roughly speaking, an \emph{expander decomposition} of a graph removes a small fraction of the edges such that each remaining connected component induces a high-conductance graph, and an \emph{expander routing} algorithm allows each node $v$ in a high-conductance graph to efficiently communicate with arbitrary $\deg(v)$ nodes and not just the local neighbors of $v$.

The all-to-all communication of expander routing allows us to emulate \clique in a high-conductance graph in \congest. Based on this idea, it was shown that listing all $k$-node subgraphs can be done in $\widetilde{O}(n^{1-2/k})$ rounds with high probability~\cite{ChangPSZ21,ChangPZ18} and deterministically~\cite{chang2024deterministic} in any graph of conductance $\phi=1/\poly \log n$ in the \congest model, matching the tight bound $\widetilde{\Theta}(n^{1-2/k})$  for $k$-clique listing that holds even in the $\clique$ model. After a sequence of works~\cite{CensorCLL21,Censor+PODC20,ChangPZ18,changS19,Eden+DISC19}, it was shown, via expander decompositions, in~\cite{CensorCLL21} that the same upper bound $\widetilde{O}(n^{1-2/k})$ not only applies to high-conductance graphs but also  applies to arbitrary graphs in the \congest model.

As we will later explain, all existing expander decomposition~\cite{ChangS20,ChangPSZ21} and routing~\cite{ChangS20,GhaffariKS17,GhaffariL2018} algorithms require at least $\mu = \deg(v) \cdot \poly \log (n)$ space per node. In order to design a triangle listing algorithm in \switch with a memory bound $\mu$ that is as small as $\Delta$, it is necessary that we modify the existing expander decomposition and routing algorithms. We will show that some existing expander decomposition and routing algorithms can be modified in such a way that yields a \emph{smooth round-space tradeoff}, which allows us to achieve $\mu = \deg(v)$ \emph{and even} $\mu = o(\deg(v))$, at the cost of slightly increasing the round complexity.

We first consider expander routing. Roughly speaking, the existing expander routing algorithms require a large space per node because these algorithms have a hierarchical structure wherein each layer has lots of routing paths needed to be stored. 
We explain the idea of our round-space tradeoff, as follows.
First of all, we will emulate a constant-degree graph in the underlying high-conductance graph by letting each node $v$ simulate $\deg(v)$ virtual nodes, and then, to save memory, our idea is to solve expander routing only for a random $1/\alpha$ fraction of the virtual nodes.
If we can do this, then repeating the above for $O(\alpha^2 \log n)$ times solves the original routing problem. The reason is that for each pair $(u,v)$ of virtual nodes such that we want to send a message from $u$ to $v$, the probability that both $u$ and $v$ are selected is $1/\alpha^2$, so repeating for $O(\alpha^2 \log n)$ times ensure that all pairs are covered.
Note that when we repeat, we forget about the routing structure we built in the previous iterations. This allows us to save space.

In order to run the existing routing algorithm on the sampled virtual nodes, we will  
 embed an $O(\log n)$-degree random graph to the selected virtual nodes
 by random walk simulation. We will show that space per node $v$ in the original graph $G$ needed to store the embedding is  $\lceil\deg(v)/\alpha \rceil \cdot O(\tmix \log n)$ with high probability,\footnote{$\tmix$ is roughly the time for a lazy random walk to reach the stationary distribution. See \Cref{sect:expander-appendix}.} so this approach allows us to save roughly a factor of $\alpha$ in memory usage at the cost of increasing the round complexity by a factor of roughly $\alpha^2$.

Next, we consider expander decomposition. Roughly, there are two different approaches in distributed expander decomposition algorithms: the nibble algorithm of Spielman and Teng~\cite{spielman2004nearly} and cut-matching games~\cite{khandekar2007cut,KRV09,RST14}. 
We will show that the distributed algorithms~\cite{ChangPSZ21,ChangPZ18,changS19} based on the nibble algorithm can be modified to yield a round-space tradeoff.

The nibble algorithm calculates the probability distribution of a lazy random walk of length $t$ starting from a node $v$ with threshold $\epsilon'$. In the calculation, probabilities below $\epsilon'$ are truncated to zero. The algorithm then ranks all the nodes by their normalized probabilities, which are random walk probabilities divided by node degrees. If the graph has a sparse cut, then for some choices of $v$, $t$, and $j$, the top $j$ nodes in the ranking form a sparse cut, as shown in~\cite{spielman2004nearly}.


The distributed algorithm in~\cite{ChangPSZ21,changS19} runs multiple instances of the nibble algorithm in parallel in such a way that each node participates in at most $O(\log n)$ instances, meaning that each node needs to maintain  $O(\log n)$ truncated random walk probabilities at any time, so during one step of the random walk probability calculation, a node $v \in V$ may receive up to $\deg(v) \cdot O(\log n)$ random walk probabilities from its neighbors, so $v$ needs to allocate a space of size  $\deg(v) \cdot O(\log n)$ to store and \mbox{process these incoming messages.}


To improve the above space complexity, we will use a strategy that is similar to the round-space tradeoff for expander routing discussed above. We will use $\alpha$ rounds to simulate one round of the distributed algorithm above in such a way that in each round each node only needs to handle incoming messages from roughly a $\frac{1}{\alpha}$ fraction of the instances of the nibble algorithm. This allows us to save roughly a factor of $\alpha$ in memory usage at the cost of increasing the round complexity by a factor of $\alpha$.

\paragraph{Upper bounds for triangle listing.} 
We show an optimal algorithm (up to sub-polynomial factors) for listing all triangles, i.e., cliques of size three, in \switch. As an overview, the algorithm of \cite{ChangPSZ21} for listing triangles in \congest works as such: compute an expander decomposition of the graph, within each expander they use expander routing to list all the triangles which use the edges of the expander or edges incident to it, then remove all the edges of the expanders from the graph and finally recurse on the edges between the expanders. While our algorithm follows the general outline for listing triangles which was established for \congest in \cite{ChangPSZ21}, the implementation of the step for listing triangles inside the clusters greatly differs due to the memory constraints of \switch. It turns out that many straightforward operations in \congest, which are required by the algorithm of \cite{ChangPSZ21}, are rather complicated to implement in \switch.

As an example, in order to deal with nodes with high degree compared to the average degree in the graph, \cite{ChangPSZ21} simply learn the entire graph in such high degree nodes. In our case, the memory bound might not allow any node to learn the entire graph, even if it has the bandwidth to receive all this information. Therefore, we show that it is possible to defer treatment of such high-degree nodes to later recursive iterations of the algorithm where the graph will become more sparse.

Further, at the heart of the listing algorithm, we are faced with a great challenge: ordering the information sent to each node such that it can perform the listing without storing too much information at once. In the \congest model, the algorithm of \cite{ChangPSZ21} reaches a state whereby every node $v$ desires to send and receive at most $O(\deg(v) \cdot n^{1/3})$ messages in order to list the triangles it is responsible to find. Using expander routing algorithms, 
this is solved directly in $\widetilde{O}(n^{1/3})$ rounds. However, in our case, we cannot afford for all the information to be directly sent to the nodes at once -- rather, we require that each node receives it in chunks of size matching the memory bound $\mu$. Further, we cannot arbitrarily chunk the data into batches of size $\mu$, as we must ensure that the edges of a given triangle are all seen at once (appear at the same batch). 

To overcome this, we begin by splitting the nodes of each cluster in the expander decomposition into $\log(n)$ buckets according to their degrees, i.e., to buckets of nodes with degrees $[1, 2)$, $[2, 4)$, $[4, 8)$, etc. Then, out of the $\log(n)$ sets we identify a set with at least $1/\log(n)$ of the total bandwidth in the graph and use only this set for the listing from now on -- this allows us to ensure that the degrees of the nodes performing the listing are very close to each other and thus we can coordinate that all these nodes will start and finish receiving each chunk of $\mu$ data at the same time. Then, in order to ensure that all the nodes which send the data are synchronized, we shuffle around the messages the nodes desire to send such that if a node is required to send more messages than it can, then other nodes in the graph lend their \mbox{bandwidth to \emph{help} this node.}

Overall, we encounter such problems in \switch compared to \congest throughout the entire implementation of the algorithm and devise unique solutions to overcome these scenarios.

\subsubsection{Simulating streaming algorithms}

We now briefly present the key techniques for Theorems~\ref{thm:streamingSimulation}--\ref{thm:composable}.

\textit{Theorems~\ref{thm:streamingSimulation}-\ref{thm:streaminSimulationLB}}. For the upper bound algorithm of \switch streaming simulation (Theorem~\ref{thm:streamingSimulation}), we propose to assign the node with the highest degree as the simulator $v$, and temporarily cache the edge information in the memory of $v$'s neighbors during the first pass. One can see that if $\mu \ge n$, the total memory of $v$'s neighbors is enough to cache all the $m \le n\cdot \Delta$ edges. During the subsequent passes, $v$ receives the cached edges directly from its neighbors to execute its streaming algorithm, reducing the round complexity from $O(n\cdot \Delta\cdot p)$ to $O(n\cdot (\Delta + p))$. 
We also prove that such an efficient simulation cannot be achieved if $\mu \le n / 4$ in the worst case (Theorem~\ref{thm:streaminSimulationLB}). 
We construct a cycle-of-cliques graph and show that the intrinsic ``bandwidth'' bottleneck of a clique (since every clique has only two outbound edges that connect to other cliques) along with the small memory results in an unavoidable $O(n \cdot \Delta)$ congestion rounds per pass.

\textit{Theorem~\ref{thm:random-order}}. In Appendix~\ref{subsec:random-order}, the key idea is to simulate the Fischer-Yates random permutation algorithm~\cite{fisher1953statistical} in the \switch model. With $\Delta$ neighbors and a total of $n\cdot \Delta$ edges, we show that the permutation problem is equivalent to selecting an arbitrary set of edges of size $n$ each time for each neighbor and re-route them to its memory, which essentially simulates Fischer-Yates algorithm. Our algorithm has two steps. First, we simulate the selection algorithm in the distributed settings with $O(n)$ memory. To that end, the simulator node $v$ decides the number of edges originally stored in each neighbor to be selected, without transmitting the actual identity of these edges, and lets the neighbors randomly draw the identity of the edges. The second is that we need to efficiently reroute each edge to its destination neighbor within $O(n)$ rounds, which means that we have to avoid congestion. We frame this rerouting as a matrix decomposition problem, where the matrix stores how many edges are to be transmitted from neighbor $i$ to neighbor $j$, and in each round, we can subtract a permutation matrix from it. By Birkhoff's Theorem, one can construct the matrix decomposition in the simulator node \mbox{$v$ with $O(\Delta^2)$ memory and thus have congestion-free scheduling.}

\textit{Theorems~\ref{thm:oneway}--\ref{thm:composable}}. We partition the data into clusters and run a streaming algorithm to summarize each cluster. The summaries are then merged in various ways, depending on their mergeability characteristics. Namely, for one-way mergeable summaries, we first divide the graph into clusters and compute summaries for every cluster, then all summaries are sent to a single node for merging. For fully-mergeable summaries, the above can be improved via a recursive process that repeatedly partitions the graph to clusters, summarizing them, and merging the results. When the summaries are also composable the merging step in the above can be done in a streaming fashion, resulting in an improved round complexity. 

\paragraph{Streaming algorithms and subgraph statistics.}
It may appear that our results regarding the simulation of streaming algorithms and subgraph listing are unrelated, however, this is not the case: we can use streaming algorithms to generate statistics regarding combinatorial structures in the network.

For example, consider the following setup: edges have colors in the range $\set{1,\dots,c}$ and we execute our triangle listing algorithm. Assume we are interested in computing some statistic regarding the number of monochromatic triangles in the graph (e.g., approximate counting, median, entropy etc.). Aggregating the triangles to a single node cannot be done in the \switch model, however, there are extremely efficient streaming algorithms to approximate these statistics. 

The applications in Section~\ref{subsec:contributions} are presented for edge streaming for simplicity; however, they readily extend to computing statistics for subgraphs, assuming that each subgraph is detected by exactly one (or constant-many) nodes in the graph.
For example, for detecting and counting the frequencies of all monochromatic triangles that appear at least $\mathfrak T\cdot \epsilon$ times, where $\mathfrak T$ is the total number of triangles, we require $O\parentheses{n^{1+o(1)}/\memory^{1/2}+{\log m\cdot\parentheses{\epsilon^{-1}\cdot \log\parentheses{\frac{\Delta\epsilon^{-1}}{\mu }}+D}}}$ rounds and $\mu\in\parentheses{(\Delta + \epsilon^{-1}),n^{4/3}}$ space.

\subsection{The roadmap of the paper}

The remaining sections are organized as follows. Our proofs regarding clique listing are elaborated in Section~\ref{sec:listing} and detailed in \Cref{appendix:triangle-listing}. Section~\ref{sec:listing} begins with lower-bound proof on clique listing (\Cref{thm: lb-listing}). For our upper-bound algorithms, we first present in \Cref{subsec:listingUBcc} as a warm-up our optimal $k$-clique listing algorithm \textit{in the all-to-all communication model} (\Cref{thm: ub-listing}). Then we expand our optimal randomized algorithm of triangle listing under the original $\mu$-CONGEST model in \Cref{appendix:triangle-listing}. Results regarding streaming simulation algorithms are elaborated in ~\Cref{sec:streaming}. In particular, we first expand on the single node $p$-pass edge-streaming problem, achieving Theorems~\ref{thm:streamingSimulation}-\ref{thm:streaminSimulationLB}. Then follows our method of simulating random-order streams (stated in Theorem~\ref{thm:random-order}) in \Cref{subsec:random-order}. Finally, as proofs of Theorems~\ref{thm:oneway}-\ref{thm:composable}, we detail our algorithms for simulating streaming algorithms that satisfy different levels of mergeability properties in~\Cref{subsec:merge-reduce}.




\section{Clique Listing in $\mu$-\congest}
\label{sec:listing}


In this section we address clique listing. We start with our lower bound proof in \Cref{subsec:listingLB}. Our corresponding optimal algorithm is presented in \Cref{subsec:listingUBcongest}, after a warm-up algorithm in an all-to-all communication setting which we give in \Cref{subsec:listingUBcc}. 

\subsection{The Lower Bound for Clique Listing in $\mu$-\congest}
\label{subsec:listingLB}

Here we prove the lower bound on clique listing. 
\ListingLB*

\begin{proof}

 We first show that any algorithm in our setting can be simulated by an algorithm that has $3\memory$ memory but can only output $k$-cliques at the end of every non-overlapping window \mbox{of $W=\lceil\memory/\ell\rceil$ rounds. }
 
 \paragraph{The simulation in the modified setup.} Consider any algorithm $A$ in the original setup. Define a new algorithm $A'$ in the new setup, which works as follows. The $3\memory$ memory at each node is split into three pieces, each of size $\memory$. One piece is called \emph{the working memory}, another piece is called \emph{the window memory}, and the other is called \emph{the storage memory}. The working memory and the window memory are initialized to the input of the node, and \mbox{the storage memory is initially empty.}

In each round $r$ of $A'$, every node $v$ collects the additional $\ell$ messages it receives in round $r$ of $A$, and stores them in the storage memory. However, if $r > W$ (recall that $W=\lceil\memory/\ell\rceil$) then the storage memory is already full. In this case, in order to store the new information for round $r$, the node deletes the information it received in round $r-W$ from the storage memory. 

Further, the node $v$ updates the working memory to be the memory that $v$ has in $A$ at the end of rounds $r$. It then prepares its messages to be sent in round $r+1$ according to algorithm $A$, according to the content of the working memory.

Finally, if $r$ is the last round in a non-overlapping window of size $W$, then the working memory is copied into the window memory.

The working memory of a node in $A'$ exactly behaves as its memory in $A$, which implies that all messages sent in $A$ are exactly those sent in $A'$, and hence the simulation is correct.

It remains to handle the output of $k$-cliques in $A'$, which can be done only in rounds that are then end of non-overlapping windows of size $W$. This is where the window memory and storage memory come into play. Specifically, if at the end of round $r$ of $A$ the node outputs some set of $k$-cliques $T$, then it outputs this set $T$ only at the end of window $w_r=\lceil W/r\rceil$ in $A'$, which is round $r'=w_r \cdot W$. Note that this is allowed in $A'$ because $r'$ is the end of a window of size $W$. The way that this is done is that before the node copies its working memory into its window memory, it takes the window memory, which holds the content of the memory of $v$ in $A$ at the end of round $r_{prev}=(w_r-1)W$, and simulates rounds $r_{prev}+1,\dots,r'$ from scratch but without sending any messages. When simulating each round $r$ in the above regime, the node $v$ can output all the $k$-cliques it outputs in round $r$ in $A$ (recall that in $A'$ the node $v$ is now in the last round of a window, as required by the setup). The reason we can re-invoke the simulation of these $W$ rounds is that we have the memory at their start stored in the window memory, and we have all of their incoming messages stored in the storage memory.

This completes the proof of the reduction to the modified setup.

\paragraph{The lower bound proof for the modified setup.} We now show the lower bound for any algorithm in the modified setup. 
The following lemma is a generalization of the case $k=3$ which is proven in~\cite{Rivin02} and used in~\cite{IzumiLG17}. 

\begin{lemma}[{See, e.g., \cite{AtseriasGM08, FischerGKO18}}]
\label{lemma:numCopies}
For every $k\geq 3$, a graph with $m$ edges contains at most $O(m^{k/2})$ copies of a $k$-clique.
\end{lemma}

For a random variable $X$ from domain $\mathcal{X}$, its \emph{entropy} is $H(X) = -\sum_{x\in\mathcal{X}}{p(x)\log(p(x))}$. The \emph{conditional entropy} of two random variables is $H(X|Y) = -\sum_{x\in\mathcal{X},y\in\mathcal{Y}}{p(x,y)\log(p(x,y)/p(y))}$, and their \emph{mutual information} is $I(X;Y) = H(X)-H(X|Y)$.

We extend here the notations given in~\cite{IzumiLG17} for $k=3$ to general values of $k$. Given a graph $G=(V,E)$, for a set $X\subseteq V$, the set $\mathcal{T}(X)$ denotes all non-ordered $k$-tuples of $X$. For any $R\subseteq \mathcal{T}$, the set $\mathcal{P}(R) \subseteq E$ denotes the set of all $e\in E$ such that $e\in t$ for some $k$-tuple $t\in R$, where $e\in t$ if and only if $t=\{i_1,\dots,i_k\}$ and $e=\{i_{\ell},i_{q}\}$, for some $1\leq \ell,q\leq k$.
Now, let $\mathcal{A}$ be a $k$-clique listing algorithm, and denote by $T_i$ the output of node $i$.

The following was proven for $k=3$.
\begin{lemma}[{\cite[Lemma 4.3]{IzumiLG17}}]
\label{lemma:infoBound}
For any node $i$, the inequality $I(E;T_i) \geq \E[|\mathcal{P}(T_i)|]$ holds.
\end{lemma}


Here, we bound the number of $k$-cliques that are output at the end of each window of size $W$, so we denote $T_i^w$ the output of node $i$ at the end of window $w$. Recall that we always address non-overlapping {windows. A similar line of proof to that of ~\Cref{lemma:infoBound} yields the following per round.}
\begin{lemma}
\label{lemma:infoBoundPerRound}
For all $k$, for every window $w$ and any node $i$, the inequality $I(E;T_i^w) \geq \E[|\mathcal{P}(T_i^w)|]$ holds.
\end{lemma}

Let $G=G_{n,1/2}$ be a random graph. Let $T_{i_{max}}$ be the node which maximizes $|T_i|$ over all nodes. The following is proven in~\cite{IzumiLG17}, using~\Cref{lemma:numCopies} for $k=3$ and~\Cref{lemma:infoBound}, where $N$ is the number of triples of nodes, i.e., $N=\Theta(n^3)$.
\begin{lemma}[{\cite[Implicit in the proof of Theorem 4.1]{IzumiLG17}}]
\label{lemma:mutualInfo}
There exists a constant $c>0$ such that 
$I(E;T_{i_{max}}) \geq c\cdot(N/n)^{2/3}$. 
\end{lemma}

We derive a similar statement, which differs from the above by addressing windows, as follows. Suppose that algorithm $\mathcal{A}$ completes within $L$ windows (with high probability, if $\mathcal{A}$ is randomized). Let $w_{max}$ and $i_{max}$ be a window number and node identifier that maximize $T_{i_{max}}^{w_{max}}$ over all windows and nodes. By~\Cref{lemma:infoBoundPerRound}, we have that $I(E;T_{i_{max}}^{w_{max}}) \geq \E[|\mathcal{P}(T_{i_{max}}^{w_{max}})|]$. 

We can generalize~\cref{lemma:mutualInfo} to hold for $k \geq 3$, and for the maximum output size per round, as follows. First, instead of $N=\Theta(n^3)$, we use the number of $k$-tuples of nodes $N_k$. Second, 
note that the $n$ term in the denominator of the bound in~\cref{lemma:infoBound} describes the size of the set over which $i_{max}$ is chosen. In our case, we maximize over all nodes and over all windows, so we need to plug $nL$ instead. Using~\cref{lemma:numCopies} for a general $k$ and~\Cref{lemma:infoBoundPerRound}, we obtain the following.

\begin{lemma}
\label{lemma:mutualInfoPerRound}
There exists a constant $c>0$ such that 
$I(E;T^{w_{max}}_{i_{max}}) \geq c\cdot(N_k/nL)^{2/k}$. 
\end{lemma}

Since the number of $k$-tuples of nodes is $N_k=\Theta(n^k)$, we get $I(E;T^{w_{max}}_{i_{max}}) \geq c\cdot(n^{k-1}/L)^{2/k}$.
Finally, denote by $\pi_i$ the transcript that node $i$ sees in algorithm $\mathcal{A}$ and by $\rho_i$ its initial state. Then~\cite{IzumiLG17} prove the following.
\begin{lemma}[{\cite[Implicit in the proof of Theorem 4.1]{IzumiLG17}}]
\label{lemma:entropy}
It holds that 
$H(\pi_{i_{max}}) \geq I(E; T_{i_{max}}) - H(\rho_{i_{max}})$.
\end{lemma}

In our case, we define $\pi_i^w$ to be the memory at node $i$ in window $w$, and $\rho_i$ its initial state. Adjusting the proof of~\Cref{lemma:entropy} gives the following.
\begin{lemma}
\label{lemma:entropyPerRound}
It holds that 
$H(\pi_{i_{max}}^{w_{max}}) \geq I(E; T_{i_{max}}^{w_{max}}) - H(\rho_{i_{max}}^{w_{max}})$. 
\end{lemma}

We now combine~\Cref{lemma:mutualInfoPerRound} and~\Cref{lemma:entropyPerRound} to conclude
 that 
 \begin{multline*}
     H(\pi_{i_{max}}^{w_{max}}) \geq I(E; T_{i_{max}}^{w_{max}}) - H(\rho_{i_{max}}^{w_{max}})\\ \geq c\cdot(n^{k-1}/L)^{2/k} - n \geq \Omega(n^{2(k-1)/k}/L^{2/k}).
 \end{multline*}

Since the entropy of the local memory lower bounds the average number of bits of local memory, this implies that $\memory\ = \Theta(3\memory) = \Omega(n^{2(k-1)/k}/L^{2/k})$, which in turn means that the number of windows $L$ is bounded from below by $\Omega(n^{k-1}/\memory^{k/2})$. Since the size of the windows is $\memory/\ell$, this implies that the number of rounds is at least $\Omega(n^{k-2}/(\memory^{k/2-1}\cdot \ell))$, as claimed.
\end{proof}

\subsection{Warm-Up: An Upper Bound for Clique Listing in All-to-All Communication with Bounded Memory}
\label{subsec:listingUBcc}
As a warm-up for our clique listing algorithms, we show how to list $k$-cliques if all-to-all communication was allowed. This would resemble the \clique model~\cite{LotkerPPP05}, in which the communication graph is a complete graph but the input graph is arbitrary, with the additional constraint that each node has only $\memory$ memory. We refer to this model as the $\mu$-\clique model, and we emphasize that we use it as a sandbox rather than an end goal, since our goal is to abstract arbitrary networks. Our approach follows the listing approach of \cite{DolevLP12}, but is memory-sensitive.


Before we delve into our main result, the nature of $\mu$-\clique does not allow us to use the typical Lenzen's routing scheme \cite{lenzenRouting} from \clique to route $\omega(n)$ messages. To recap, Lenzen's routing scheme in \clique states that if every node desires to send and receive at most $O(n)$ messages, then all of these messages can be delivered in $O(1)$ rounds. For any $x \geq 1$, this can be generalized to the statement that if every node desires to send and receive $O(x\cdot n)$ messages, then this can be done in $O(x)$ rounds. However, in our case, $x\cdot n$ might be greater than $\mu$, in which case we would want to somehow batch the data so that each node would receive $\mu$ messages, process those messages, and then receive the next $\mu$ messages. While we show how to balance the receiving part (i.e. that each node receives messages in chunks of size $\mu$), it becomes challenging to show that each node sends messages in chunks of size $\mu$. Thus, we turn to the \mbox{following lemma from \cite{loadBalancedRouting}.}

\begin{lemma}[Lemma 9 in \cite{loadBalancedRouting}]\label{load-balanced-routing-clique}
Any routing instance in which every node $v$ is
the target of up to $O(n)$ messages, and $v$ locally computes the messages it desires to send
based on at most $|R(v)| = O(n \log n)$ bits which it knows, can be performed in $O(1)$ rounds of \clique.
\end{lemma}

In our case, the number of bits based on which a node needs to generate its outgoing messages is bounded by the size of its local memory as well as its incoming messages. Note that while the lemma is proven in \cite{loadBalancedRouting} for the \clique model, it is possible to observe in their proof that the algorithm uses $\Theta(n)$ memory, and as such as long as $\mu \geq n$, we can use this lemma, to get the following lemma in $\mu$-\clique.

\begin{lemma}\label{load-balanced-routing-clique-switch}
Given $\mu \geq n$, any routing instance in which every node $v$ is
the target of up to $O(n)$ messages, and $v$ locally computes the messages it desires to send
based on at most $|R(v)| = O(n \log n)$ bits, can be performed in $O(1)$ rounds of $\mu$-\clique.
\end{lemma}

We are now ready to prove our main result for this warm-up.


\begin{restatable}{theorem}{ListingUBcc}
\label{thm: ub-listing}
Given $n \le \memory \le n^{2-2/k}$, there exists a deterministic algorithm  for listing all the $k$-node subgraphs within $O(n^{k-2}/\memory^{k/2 - 1})$ rounds in the $\mu$-\clique model. 
\end{restatable}

\begin{proof}
Our first step is to split the graph $G$ into $n^{1/k}$ sets of subgraphs, each of size $n^{(k-1)/k}$ nodes, and assign every $k$-tuple of the sets to a unique \emph{master node} for listing $k$-cliques. Thus, each node is a master node that is responsible for $k$-cliques consisting of only $kn^{(k-1)/k}$ nodes. To list the cliques at each master node, we use the following approach.

\begin{definition}
An \emph{$(a, b, c)$ subset cover} is a sequence of sets $(S_1, \ldots, S_z)$ for some value $z$, such that (i) for all $1\leq i \leq z$ we have $S_i \subseteq \{1, \ldots, a\}$ and $|S_i|=b$, and (ii) for every set of $c$ elements $L = \{l_1, \ldots, l_c\} \subseteq \{1, \ldots, a\}$, there exists a set $S_i$ such that $L \subseteq S_i$.
\end{definition}

To construct a subset cover, partition the $a$ elements into sets of size $b/c$, denoted by $T_1, T_2, \ldots, T_{ac/b}$. Then, for each $c$-tuple $(T_{i_1}, \ldots, T_{i_c})$ of these sets, let $S_i = \cup_j T_{i_j}$. There are in total $z=(ac/b)^c$ tuples, and it is easy to verify that $(S_1, \ldots, S_z)$ satisfies the requirements of a subset cover.  

The way we use $(a,b,c)$ subset covers is as follows. For each master node, we solve the $k$-clique listing problem for its $kn^{(k-1)/k}$ nodes by using a construction of an $(kn^{(k-1)/k}, \sqrt{\memory}, k)$ subset cover. This construction is predetermined and requires no communication. Then, for $z$ iterations, in iteration $i$, all the edges $\{u, v\}$ such that $u, v \in S_i$ are sent by their endpoints to the master node, and thus it can list all the $k$-cliques in $S_i$. 

\paragraph{Analysis:}
First, it is easy to see that this satisfies the $O(\memory)$ memory constraint at the master node because we only send to it $(\sqrt{\memory})^2 = \memory$ edges during an iteration.

It remains to bound the number of rounds. Note that we have $z=(ac/b)^c$ in our construction, and by plugging in the values for $a,b,c$, we have $((kn^{(k-1)/k})k/\memory^{1/2})^k = (k^{2k})(n^{k-1}/\memory^{k/2}) = O(n^{k-1}/\memory^{k/2})$ iterations.  



For $\Theta(n) \le \memory \le \Theta(n^2)$, we use $\memory/n$ rounds to process each subset $S_i$. Indeed, using Lemma~\ref{load-balanced-routing-clique-switch}, since each master node receives $O(\mu)$ messages, the round complexity for processing one subset is $O(\memory / n)$. The round complexity is thus $O((n^{k-1}/\memory^{k/2}) \cdot (n/\memory)) = O(n^{k-2}/\memory^{k/2 - 1})$.

Note that if $\memory > \Theta(n^{2-2/k})$, then the memory is large enough to store all the edges in the assigned subgraph component (with $O(n^{(k-1)/k})$ nodes and $O(n^{2-2/k})$ edges). In that case, the only communication is to send all the $O(n^{2-2/k})$ edges to the master node, thereby requiring only $O(n^{1-2/k})$ rounds, due to Lenzen's routing scheme~\cite{Lenzen13}, which allows sending and receiving $O(n)$ messages by each node in $O(1)$ rounds. 
\end{proof}

\subsection{The Upper Bound for Triangle Listing in $\mu$-\congest}
\label{subsec:listingUBcongest}

We are now derive our \switch algorithm for clique listing.

\ListingUBcongest*

In the \congest model, this result is shown in \cite{ChangPSZ21}. We desire to replicate their results in the \switch model, however, many parts of their algorithm break due to the memory constraints in \switch. Thus, we show a variety of new techniques to overcome these challenges.

We caution the reader that as it is crucial in the \switch model to not exceed the memory limit per node, the following proofs are extremely detailed in order to ensure correctness. Therefore, they are deferred to \Cref{appendix:triangle-listing}.

\section{Simulation of Streaming Algorithms}
\label{sec:streaming}
In the centralized setting, handling memory constraints is typically done using {streaming} algorithms.
Such algorithms make $p\ge 1$ passes over the data (e.g., graph edges) and are allowed to store only a limited amount $M$ of information at any point in time. 
In this section, we explore how such algorithms can be leveraged for designing \switch algorithms, by simulating them at the nodes of the network.

\paragraph{The edge-streaming model.} A special case of the streaming model, that well fits our motivation, is called edge-streaming.
Since the number of edges in the graph might be too large to fit in memory, algorithms may only store $M=o(n^2)$ words. 
The edges of the graph are processed as a stream of information, perhaps in several passes, and according to the setting, their order could be arbitrary, adversarial, or random.
When it is clear from context, we write ``streaming'' algorithm (instead of edge-streaming). 

A special case of the above is the popular \emph{semi-streaming} model, where $M=O(n \cdot  \text{polylog } n)$ \cite{feigenbaum2005graph}.
As some tasks have superior guarantees (use fewer passes or obtain better approximations) when edges arrive in random order, we discuss how to shuffle the edges distributedly to emulate random order in \cref{subsec:random-order}. 

We start by assuming a single \emph{simulator node} which runs the edge-streaming algorithm. 
The vanilla approach of collecting all edges at each pass to a single simulating node takes $O(n\cdot \Delta\cdot p)$ rounds, where $\Delta$ is the maximum degree of the graph. Instead, we propose a distributed \switch simulator of the edge-streaming model that completes in just $O(n\cdot (\Delta+p))$ rounds when $\mu=\Omega(n)$. In contrast, we show that any single-node simulation of a $p$-pass streaming algorithm requires $\Omega(n\cdot \Delta\cdot p)$ rounds if $\memory = o(n)$. Note that the benefit of an $o(n)$ simulator would have been small anyhow {since even elementary graph properties like connectivity are intractable with less space~\cite{feigenbaum2005graph}.}

To adapt edge-streaming algorithms to \switch, we need to consider the overhead of sending the edges to the simulator node (which runs the streaming algorithm). 
We show that by picking the largest degree node as the simulator node, we can avoid collecting the edges from scratch at each pass.
%
%
%
The following theorem summarizes our result.


\StreamingSimulation*

\begin{proof}
Let $v$ be the node with the largest degree in the graph, i.e., $deg(v)=\Delta$. Since the number of edges is $|E|\le n\cdot \Delta$, we have that the neighbors of $v$ have enough memory to (collectively) store all the edges of the graph, which can be achieved by storing $n$ edges at each neighbor.
Our simulator then first forwards all the edges to $v$'s neighbors, after which they stream the edges one by one for each of the $p$ passes. Observe that it requires at most $O(n\cdot \Delta)$ rounds to collect the edges at the neighbors (by forwarding edges to the simulator node taking $O(n\cdot\Delta)$ rounds and then routing them to the neighbors), and another $O(n\cdot p)$ rounds to send them to the simulator, giving a total of $O(n\cdot (\Delta+p))$ rounds. 
Note that this improves over the simple solution of collecting the edges from scratch at each round, which would take $O(n\cdot \Delta\cdot p)$ rounds. 
%
\end{proof}


In contrast, we show that using sublinear memory does not allow efficient simulation, as stated in the following theorem.

\StreamingSimulationLB*

\begin{proof}
    Consider a graph with $n/(\Delta-1)$ cliques of size $\Delta-1$ that are connected in a cycle. Formally, the set of nodes is 
    $V = C_0\cup\ldots\cup C_{n/(\Delta-1)-1}$, where $C_i=\set{v_{i,0},\ldots,v_{i,\Delta-2}}$, and the set of edges is $E=\set{\set{v_{i,j},v_{i,k}}\mid i\in\brackets{n/(\Delta-1)}, j,k\in[\Delta-1], j\neq k} \cup \set{\set{v_{i,0},v_{(i+1\mod (\Delta-1)),0}}}$.
    Suppose we choose node $v$ within the $i$-th clique as the simulator node. The neighbors of $v$ consist of $\Delta - 2$ nodes within the clique and perhaps two additional nodes from the $(i\pm1)$-th cliques. 
    The nodes in $v$'s clique can only store at most $\mu \cdot (\Delta-1) < n\cdot\Delta/4$ edges. Therefore, the remaining $\Omega(n\cdot \Delta)$ edges of the graph are stored outside the $i$-th clique. Since only two edges connect $C_i$ to the rest of the graph, sending these $\Omega(n\cdot \Delta)$ edges require $\Omega(n\cdot \Delta)$ rounds, and this process has to repeat for each of the $p$~passes.
\end{proof}

\subsection{Mergeable Streaming Simulations}
\label{sec: mergeable streaming}
\label{subsec:merge-reduce}
In this section, we consider general streaming algorithms (not necessarily edge-streaming). In this model, we assume that every node in the graph has some initial input data $I_v = a_{v,1}, a_{v,2}, \dots, a_{v,t_v}$ where $a_{v,i} \in U$ for a universe of size $|U|=n^{O(1)}$, and our goal is to execute a streaming algorithm over the multi-set $\mathcal I=\bigcup_{v\in V}I_v$. Intuitively, a streaming algorithm on a stream $I$ can be viewed as maintaining a \emph{summary} $S(I)$ of the data $I$. The summary is simply the memory of the streaming algorithm after processing the stream. 

The algorithm in the proof of \Cref{thm:streamingSimulation} leverages the memory of neighboring nodes to reduce congestion for the simulator node. However, such an optimization only starts to reduce the round complexity from the second pass onward; during the first pass, all the input data $\mathcal I$ across the graph still has to be transmitted to the simulator node $v$, incurring a $\Theta(n \cdot \Delta)$ round complexity in the worst case, which in general cannot be improved. However, for some streaming algorithms satisfying \emph{mergeability}, we can further reduce the round complexity. 
Intuitively, mergeability allows us to run streaming algorithms on subsets of the data and then merge their summaries to enable queries about the complete data.

Our key idea is to partition the graph into several disjoint subgraphs, within each of which a streaming simulator is established to compute a solution for this subgraph. Then mergeability allows us to directly join solutions from different subgraphs together.

\subsubsection{One-way mergeable summaries}
Let us formally define this notion of mergeability, slightly modified from~\cite{agarwal2013mergeable}:
\begin{definition}   
A summary $S(I)$ of size $M$ is one-way mergeable if there exist two algorithms $A_1$ and $A_2$ such that, (1) given any $I$, $A_2$ creates a summary $S(I)$; (2) given any $M$-sized $S(I_2)$ produced by $A_2$ and any $M$-sized $S(I_1)$ produced by $A_1$ or $A_2$, $A_1$ builds an $M$-sized merged summary $S(I_1\cup I_2)$.
\end{definition}

We analyze the distributed complexity of simulating a one-way mergeable streaming algorithm.

\oneway*

\begin{proof}
    Our goal is to divide the graph into disjoint clusters, run a streaming algorithm in each cluster, and send the summaries to a leader node that merges them.

    We start by dividing the graph into disjoint clusters, each containing $s=\sqrt{|\mathcal I|\cdot M}$ elements. To construct the clusters, we first build a BFS tree. We define the algorithm recursively, starting from the leaves. 
    Every leaf $v$ does the following. If $t_v>s$, then $v$ is a cluster and $v$ its parent $0$. Otherwise, $v$ is not a cluster and it sends $t_v$ to its parent. Every non-leaf node $v$ that receives the following values $\set{\tau_{u_1},\ldots,\tau_{u_t}}$ for every child node ($u_1, \ldots ,u_t$) and creates clusters by greedily partitioning them into clusters $\set{u_i,u_{i+1},\ldots,u_j}$ for which $\sum_{i\le z\le j} \tau_z\in[s,3s]$. That is, the first cluster would be $\set{u_1,\ldots,u_i}$ for the smallest $i$ for which $\sum_{1\le z\le i} \tau_z\in[s,3s]$; the next clusters will start from $i+1$, etc. Let $u_\ell,\ldots,u_t$ be the nodes that are left unclustered (as $\sum_{\ell\le z\le t} \tau_z< s$).
    If $t_v\ge s$, then $\set{v}$ is a cluster and $\set{u_\ell,\ldots,u_t}$ is another cluster, and the message $v$ sends its parent is $0$.
    Otherwise (if $t_v< s$), $v$ sends its parent $t_v+\sum_{\ell\le z\le t} \tau_z$. Notice that the message $v$ sends is smaller than $2s$, and thus clusters must be of size at most $3s$.

    For each cluster $\set{u_i,\ldots,u_j}$, we pick the lowest ID node as the cluster's leader, which then uses $A_2$ to summarize the data (which can be of size at most $3s$) of all nodes in the union of the subtrees rooted at $\set{u_i,\ldots,u_j}$.  
     Note that although clusters are not necessarily connected, $u_i,\ldots,u_j$ may communicate with each other via $v$. This may cause congestion on the edge $(v, u_i)$. However, within a cluster, we simply pipeline all elements to the leader; therefore, the leader still receives all elements within $3s$ rounds. 
     

    The summaries are then communicated to a global leader node that runs $A_1$ to merge the summaries. There are at most $\min \set{n, 2|\mathcal{I}|/s}= O(\min \set{n,\sqrt{|\mathcal{I}|/M}})$ clusters. This is because either a cluster is of size at least $s$, or its parent $v$ has $t_v>s$ elements. Since there can be at most $|\mathcal{I}|/s$ such $v$ vertices, the overall number of clusters is no larger than $2|\mathcal{I}|/s$.
    This means that sending all summaries to the global leader can be done in $O\parentheses{\min \set{n M,\sqrt{|\mathcal{I}|\cdot M}} + D}$ rounds.
    Finally, the global leader broadcasts the result to all nodes.
\end{proof}

\subsubsection{Fully mergeable summaries}
We formally define this notion of mergeability by adapting it from~\cite{agarwal2013mergeable}. Notice that unlike in the one-way mergeable case, both merged summaries can be created by either $A_1$ or $A_2$.
\begin{definition} 
A summary $S(I)$ of size $M$ is fully-mergeable if there exist two algorithms $A_1$ and $A_2$ such that, (1) given any $I$, $A_2$ creates a summary $S(I)$; (2) given any $M$-sized summaries $S(I_1), S(I_2)$ produced by \underline{$A_1$ or $A_2$}, $A_1$ builds an $M$-sized merged summary $S(I_1\cup I_2)$.
\end{definition}

\fullym*

\begin{proof}
To leverage mergeability, we recursively break the problem into smaller ones. Namely, we start by building a BFS tree $T$ rooted at an arbitrary node $v$. We want to find a node $u$ such that all subtrees $\mathcal{T} = \set{T_1,\dots ,T_\ell}$ of $T\setminus\set{u}$ contain at most $|\mathcal I|/2$ information. Initially, we set $u=v$. If $u$ is not such a node, there must be a subtree with more than $|\mathcal I|/2$ information, so we set $u$ to the root of this subtree. We repeat this process until the desired $u$ is found. This process is guaranteed to terminate within $O(D)$ rounds.

We recurse this process for all $\mathcal{T}$ where $|\mathcal I|_i$ in the recursion step for $T_i$ is the sum of all inputs for nodes in the subtree. That is, we have $|\mathcal I|=I_u+\sum_{i=1}^\ell |\mathcal I|_i$. When the subtree is a singleton we simply run the summary on the data of the node.

Let us now describe how to merge summaries when returning from the recursion for some node $u$ and subtrees $\mathcal{T} = \set{T_1,\dots ,T_\ell}$. We propagate the summary of $T_i$ to its root in $O(M+D)$ rounds. Then we run a summary over $I_u$ and start merging the summaries of $\mathcal{T}$ in a pairwise fashion. That is, we match elements 
in $\mathcal{T}$ in a pairwise fashion and use $u$ to forward summaries between matched subtrees. When a subtree receives the summary data of its matched subtree, it merges the two summaries. We repeat this until the total size of the summaries is at most $\mu$. At this point we sent all summaries to $u$ for merging. Note that while we are using just $O(M)$ memory in the neighboring nodes, it seems that one cannot leverage the additional memory to accelerate the merging process. This is because if an algorithm merges the summaries of, e.g., $T_1,T_2,T_3$ together at a node in $T_1$, we have $2M$ words that need to be sent over the edge between $u$ and $T_1$, thus requiring $2M$ rounds. Instead, we could merge $T_3$'s summary into $T_2$ in $M$ rounds and then merge the result into $T_1$ in another $M$ rounds. Further, merging summaries at $u$ throughout the pairwise merging process does not asymptotically decrease the round complexity. 

Notice that we have $O(\log\frac{\Delta}{\mu / M})$ iterations of pairwise-merging, which means that the entire merging operation takes $O(M\cdot \log\frac{\Delta}{\mu / M})$ rounds.
Finally, observe that in each recursion step, we at least halve the amount of information in each subtree, thus the recursion ends after no more than $\log |\mathcal I|$ steps.
Put together, this give a round complexity of $O\parentheses{\log |\mathcal I|\cdot\parentheses{M\cdot \log\frac{\Delta}{\mu / M}+D}}$. Finally, note that by summarizing the information of each node $v$ for which $|I_v|>M$ before starting the recursion, we can refine the complexity to $O\parentheses{\parentheses{\log\min\set{n\cdot M, |\mathcal I|}}\cdot\parentheses{M\cdot \log\frac{\Delta}{\mu / M}+D}}$.
\end{proof}
\subsubsection{Composable summaries}
We present a third notation of mergeability that allows summaries to be merged on the fly, without needing to store them in memory first. That is, a node can receive summaries from its neighbors and merge them using $M$ memory in $O(M)$ rounds.
\begin{definition} 
A summary $S(I)$ of size $M$ is fully-mergeable and composable if there exist two algorithms $A_1$ and $A_2$ such that, (1) given any $I$, $A_2$ creates a summary $S(I)$; (2) given any $M$-sized summaries $S(I_1),\ldots S(I_\ell)$ produced by {$A_1$ or $A_2$}, there exists a function $g$ such that $g(S(I^j)) = \set{w_i^j}_{i=1}^M$ and $A_1$ can process $\set{\set{w_i^j \mid j\in[\ell]}\mid i\in[M]}$ in a streaming fasion using $M$ memory to compute $A_1(\bigcup_{j=1}^\ell S(I^j))$.
%
\end{definition}

Using the above definition we can improve over the running time of the algorithm in Theorem~\ref{thm:fullym}.

\composable*
\begin{proof}
We repeat the same recursive process as in the proof of Theorem~\ref{thm:fullym} but merge all summaries of $u$'s subtrees directly at $u$. This means that each recursion step requires $O(M+D)$ rounds, for a total $O\parentheses{\parentheses{\log\min\set{n\cdot M, |\mathcal I|}}\cdot\parentheses{M+D}}$.
\end{proof}

\section{Discussion and future works}
\label{sec:discussion}

Some of our algorithms are provably tight, such as the triangle listing algorithm. We believe that listing larger cliques can also be done within the round complexity given in our lower bound, but this requires additional machinery. However, we do not rule out the possibility of improved algorithms in \switch for other tasks, and certainly there are tasks that are not addressed in our work -- all of these are intriguing future research directions. 

Moreover, once memory limitations are imposed, many distributed algorithms in other models do not work as is. For example, in the 
\local model \cite{Linial92}, where there are no bandwidth restrictions, there exist fast $(1\pm \epsilon)$-approximation algorithms for all covering/packing integer programming problems, such as minimum dominating set, and more~\cite{chang2023complexity,GhaffariKM17}. These algorithms are based on a ball-carving approach which inherently collects large neighborhoods, and hence if the memory is bounded, new approaches must be incorporated into this general algorithmic structure.

One can also consider the case where $\memory\leq\Delta$. This requires a notion of an \emph{adversary}. 
In the standard \congest model, the typical setting is that an  adversary (knowing the algorithm) selects a worst-case input graph (network topology). In \switch with very low memory, the order of messages received from neighbors can also affect certain algorithms.
An \emph{oblivious} adversary does this without knowing the random coin flips of the algorithm, while an \emph{adaptive} adversary may choose the order of arriving messages at round $r$ after knowing all of the random bits used by the algorithm up to round $r$. This opens additional research \mbox{directions for future work}.

\update{Finally, in real-world networking scenarios, the behaviors of network switches such as Tofino~\cite{Tofino} are characterized in a more complex way, which the \switch model does not capture. For example, they are inherently limited in their total computational capacity for processing each incoming message (packets), and do not natively support for-loops~\cite{ben2018efficient}. While packets can be recirculated~\cite{tofino-spec} to emulate unlimited \textit{per-packet computation}, the total packet-processing computation over all incoming packets of a switch remains a constant value. To make the model more faithful to the switch's computational capacity, we may want to further bound the per-round computational complexity by $O(n)$. Another example is that programmable switches feature limited support for complex arithmetic primitives. For instance, Tofino~\cite{Tofino} supports constant-time ternary lookup via its TCAM memory~\cite{hucaby2004ccnp}, but the TCAM memory capacity is extremely limited. Moreover, complicated arithmetic operations such as logarithms and divisions are not intrinsically supported, and these operations can only be approximated, which again consumes TCAM memory~\cite{ben2020pint}. We claim that extending \switch to model these practical switch hardware limitations is interesting to study, and we leave it as future work.}

\begin{acks}
Keren Censor-Hillel is supported in part by the Israel Science Foundation (grant 529/23).
Yi-Jun Chang is supported in part by the Ministry of Education, Singapore, under its Academic Research Fund Tier 1 (24-1323-A0001). Ran Ben Basat is supported in part by a gift funding from Huawei. Gregory Schwartzman is supported by the following research grants: KAKENHI 21K17703, KAKENHI 25K00370, JST ASPIRE JPMJAP2302 and JST CRONOS Japan Grant Number JPMJCS24K2.
\end{acks}


\bibliographystyle{plainurl}
\bibliography{reference}

\appendix
\section{Distributed expander decomposition and routing}\label{sect:expander-appendix}

In this section, we examine the space complexity of distributed algorithms for expander decomposition and expander routing, which are useful building blocks in designing distributed graph algorithms, particularly for distributed subgraph detection problems. 
All algorithms in this section are in the \congest model or in the \switch model whenever a memory size bound $\mu$  is specified. In this section, we allow the possibility that  the memory size $\mu$ of a node $v$ is a function of $\deg(v)$. The round-space tradeoffs presented in this section allow $\mu$ to go below $\deg(v)$.

We begin with some basic graph terminology. Let $G=(V,E)$ be any graph, For any set of nodes $S \subseteq V$, we define the \emph{volume} $\vol(S)$ of $S$ as $\sum_{v \in S} \deg(v)$, and we define the \emph{conductance} $\Phi(S)$ of $S$ as the number of edges crossing $S$ and $V \setminus S$ divided by $\min\{\vol(S), \vol(V \setminus S)\}$. Intuitively, the conductance of $S$ is a measure of the size of the cut $(S, V \setminus S)$ that takes the volume of $S$ and $V \setminus S$ into consideration. 
We define the conductance $\Phi(G)$ of a graph $G$ as the minimum value of $\Phi(S)$, ranging over all $S\subseteq V$ such that $S \neq \emptyset$ and $S \neq V$.

We say that a graph $G=(V,E)$ is an \emph{$\phi$-expander} if $\Phi(G) \geq \phi$. An \emph{$(\epsilon, \phi)$-expander decomposition} is a partition $V=V_1 \cup V_2 \cup \cdots \cup V_s$ satisfying the following two conditions.
\begin{itemize}
\item The subgraph $G[V_i]$ induced by each part $V_i$ with $|V_i| > 1$ is a $\phi$-expander.
\item The number of edges crossing different parts is at most $\epsilon |E|$.
\end{itemize}

There are several natural variants of the definition of expander decomposition. The distributed algorithms in~\cite{changS19,ChangPSZ21} computes a stronger variant of expander decomposition that requires not only $G[V_i]$ but also $G\{V_i\}$ to be a $\phi$-expander. Here $G\{V_i\}$ is defined as the result of adding $\deg_{G}(v) - \deg_{G[V_i]}(v)$ self-loops to each $v \in V_i$ in $G[V_i]$, and each self-loop contributes one to the calculation of degree and volume. The distributed algorithms in~\cite{ChangS20} computes a weaker variant of expander decomposition where each $\phi$-expander in the decomposition is not required to be a node-induced subgraph: Their algorithm removes at most $\epsilon$ fraction of the edges such that each remaining connected component with more than one node is a $\phi$-expander. 

In the subsequent discussion, we call these two variants \emph{strong} and \emph{weak} $(\epsilon, \phi)$-expander decompositions. For $\phi = 1/\poly(1/\epsilon, \log n)$, a weak $(\epsilon, \phi)$-expander decomposition can be computed in $\poly(1/\epsilon, \log n)$ rounds with high probability~\cite{ChangS20}, and a strong $(\epsilon, \phi)$-expander decomposition can be computed in $\poly(1/\epsilon) \cdot n^c$ rounds with high probability, with an arbitrarily small constant exponent $c > 0$~\cite{changS19,ChangPSZ21}.

Intuitively, expander decomposition is useful because it allows us to focus on $\phi$-expanders, where we may apply an efficient \emph{expander routing} algorithm to let each node $v$ communicate with arbitrary $\deg(v)$ nodes and not just the local neighbors of $v$. 

Formally, expander routing is a routing problem where each node $v$ is the source and the sink for at most $L \cdot \deg(v)$ messages of $O(\log n)$ bits. In a $\phi$-expander, expander routing can be solved with high probability in $L \cdot \poly(1/\phi) \cdot 2^{O\left(\sqrt{\log n}\right)}$ rounds~\cite{GhaffariKS17,GhaffariL2018} or $L \cdot \poly(1/\phi, \log n)$ rounds after a preprocessing step that costs $\poly(1/\phi) \cdot O(n^c)$ rounds, with an arbitrarily small constant exponent $c > 0$~\cite{changS19,ChangPSZ21}.

\subsection{Existing expander routing algorithms}\label{sect:routing-algo}
In a \emph{lazy random walk} in a graph, in each  step, the walk stays at the current node $v$ with probability $1/2$ and moves to a uniformly random neighbor of $v$ with probability $1/2$. 
Same as~\cite{GhaffariKS17,GhaffariL2018}, we define the \emph{mixing time} $\tmix$ of a graph $G=(V,E)$ as the smallest integer such that for any $u \in V$, $v \in V$, and $t \geq \tmix$, we have 
$$\left|p_v^{t}(u) - \frac{\deg(u)}{\sum_{w \in V} \deg(w)}\right| \leq \frac{1}{|V|} \cdot \frac{\deg(u)}{\sum_{w \in V} \deg(w)},$$
where $p_v^{t}(u)$ is the probability that a lazy random walk starting at $v$ is at $u$ after $t$ steps. It is known~\cite{JerrumS89} that $\tmix = O\left(\frac{\log n}{\phi^2}\right)$ for any $\phi$-expander.

\subsubsection{The Ghaffari--Kuhn--Su algorithm}

The algorithm of Ghaffari, Kuhn, and Su~\cite{GhaffariKS17} solves the expander routing problem in $L \cdot \tmix \cdot 2^{O\left(\sqrt{\log n \log \log n}\right)}$ rounds with high probability. We briefly review their algorithm and verify that the algorithm costs $\mu = \deg(v) \cdot \left(L + \tmix \cdot \poly \log n\right)$ space per node $v \in V$.

The algorithm solves the expander routing problem in a recursive manner. In the first step, the algorithm simulates a random graph of $2|E|$ nodes of degree $O(\log n)$ by letting each node $v \in V$ locally simulate $\deg(v)$ virtual nodes. The edges in the random graph are computed by performing lazy random walks of length $2 \tmix$.  The embedding of the random graph has congestion $c = \tmix \cdot O(\log n)$ and dilation $d = 2\tmix$ in the sense that each edge in the random graph is simulated as a path of length $d$ where each edge in the original graph is used for at most $c$ times. The storage of the random graph already costs $\deg(v) \cdot \tmix \cdot  O(\log n)$ space per node $v \in V$.

There is one subtle issue that was discussed in a remark in~\cite{ChangPZ18,ChangPSZ21}: We need to be able to translate a routing request in the original graph to a routing request in the random graph. One natural approach is the following~\cite{ChangPZ18,ChangPSZ21}. Suppose we need to send a message from $u$ to $v$ in the original graph. Then the new source of the message will be a virtual node of $u$ sampled uniformly at random from all the virtual nodes corresponding to $u$, and the new destination of the message will be a virtual node of $v$ sampled uniformly at random from all the virtual nodes corresponding to $v$. Doing so guarantees that the maximum load per node in the random graph will be $O(L + \log n)$ with high probability.

For the next level of the recursion, the algorithm partitions the $2|E|$ virtual nodes in the random graph into $\beta$ parts $V_1 \cup V_2 \cup \cdots \cup V_{\beta}$, and then the algorithm computes a simultaneous embedding of $\beta$ random graphs of degree $O(\log n)$ to each part, with polylogarithmic congestion and dilation. The total number of levels is $k = O(\log_{\beta} n)$. In the implementation in~\cite{GhaffariKS17}, they set $\beta = 2^{O\left(\sqrt{\log n \log \log n}\right)}$, so $k = O\left(\sqrt{\frac{\log n}{\log \log n}}\right)$. Similarly, for each level of recursion, the cost of storing the embedding is $\deg(v) \cdot  \poly\log n$ space per node $v \in V$ in the original graph.

To transmit the messages across different parts, the approach taken in the Ghaffari--Kuhn--Su algorithm is as follows. If $v \in V_i$ has a message whose destination is in $V_j$ with $j\neq i$, then $v$ needs to first route the message to another node $u \in V_i$ that is adjacent to a node in $V_j$, and then the message can be transmitted to $V_j$ via an edge incident to $u$. Such a node $u$ is called a portal of $v$ for $V_j$. The portals can be computed by lazy random walks within the random graph embedded to $V_i$. The information about portals does not need to be stored.

In the last level of recursion where the number of nodes is already $2^{O\left(\sqrt{\log n \log \log n}\right)}$, partitioning the node set into $\beta = 2^{O\left(\sqrt{\log n \log \log n}\right)}$ is not helpful and we need to solve the routing problem without recursion. The algorithm for the last level of recursion was not described in~\cite{GhaffariKS17}. One natural approach is to simply deliver all the messages by running lazy random walks until all the messages are delivered. The round complexity of the algorithm is $2^{O\left(\sqrt{\log n \log \log n}\right)}$.

To summarize, the overhead in the space complexity is an additive term of $\deg(v) \cdot  \poly\log n$ per level of recursion, except for the first level where the cost is $\deg(v) \cdot \tmix \cdot O(\log n)$ for each node $v \in V$, so the overall memory complexity is  $\mu = \deg(v) \left(O(L) + \tmix \cdot \poly \log n\right)$ space per node $v \in V$. The extra $\deg(v) \cdot O(L)$ term is the memory cost for actually routing the messages. Note that we can always perform lazy random walks to balance the load per node to ensure that at any time the load per virtual node is $O(L + \log n)$ with high probability, so the load per node $v$ in the original graph is  $\deg(v)  \cdot O(L + \log n)$ with high probability throughout the algorithm.

\subsubsection{The Ghaffari--Li algorithm}
Subsequent to~\cite{GhaffariKS17}, Ghaffari and Li~\cite{GhaffariL2018} presented an improved algorithm that costs $L \cdot \tmix \cdot 2^{O\left(\sqrt{\log n }\right)}$ rounds.
The improvement in the round complexity comes from a more efficient emulation of random graphs in each level, with only constant  congestion and dilation, so they are able to set $\beta = 2^{O\left(\sqrt{\log n }\right)}$ and $k = O\left(\sqrt{\log n }\right)$. The improvement is achieved by reusing some randomness used in the previous level.

The improvement comes at a cost that in each level, the algorithm needs to emulate random graphs of degree $2^{O\left(\sqrt{\log n }\right)}$ and not $O(\log n)$. Storing the information about the incident edges and the corresponding path in the underlying graph already costs an overhead of $2^{O\left(\sqrt{\log n }\right)}$ in the space complexity. In particular, even the cost of building the first-level random graph is already $\deg(v) \cdot \tmix \cdot 2^{O\left(\sqrt{\log n }\right)}$ space per node $v\in V$. In its current form, the Ghaffari--Li is not as space-efficient as the Ghaffari--Kuhn--Su algorithm.
The space complexity of the Ghaffari--Li algorithm is $\mu = \deg(v) \cdot \left( O(L) + \tmix \cdot 2^{O\left(\sqrt{\log n }\right)}\right)$ per node $v\in V$.

\subsubsection{Polylogarithmic-round algorithms with preprocessing}

As observed in~\cite{ChangPSZ21,changS19}, the Ghaffari--Kuhn--Su algorithm can be turned into one that costs small polynomial rounds in preprocessing and then polylogarithmic rounds for each subsequent routing task. This is achieved by setting $\beta = O(n^c)$ for any arbitrarily small constant $c > 0$, so the number of levels can be a constant, meaning that the overall congestion and dilation for routing the messages is only $\poly \log n$. Specifically, the cost of the preprocessing is  $\tmix \cdot n^{O(c)}$ rounds and the cost for routing is $L \cdot \tmix \cdot \poly \log n$ rounds.
As we discuss below, this approach inherently needs polynomial space per node. We will show that the memory cost is $\mu = \deg(v) \cdot \left(O(L) +  \tmix \cdot \poly \log n +  n^{O(c)}\right)$ per node $v\in V$.

We briefly explain the source of the additive polynomial term $n^{O(c)}$.
The first reason is that now we need to store all $\beta-1$ portals per virtual node during the preprocessing because we do not know which portals will be used in each of the subsequent routing requests. This already adds a polynomial space overhead, as $\beta = O(n^c)$.

The second reason is that the last level of the recursion does not seem to have an implementation that is efficient in both 
 round and space complexities. The method discussed earlier is space-efficient but it costs super-linear rounds, which we cannot afford here as we aim for a polylogarithmic-round routing algorithm. The implementation of the last step was not discussed in~\cite{ChangPSZ21,changS19}. Below we describe a possible implementation. We locally store a $2 \tmix$-length random walk from $u$ to $v$ locally at $u$, for each pair $(u,v)$, where  $\tmix = O(\log n)$ for the random graph constructed in the Ghaffari--Kuhn--Su algorithm. To use this information in the routing algorithm, for each message needed to be sent from $u$ to $v$, we simply use the pre-computed walk from $u$ to $v$, so  the routing algorithm is round-efficient in that it costs only $\poly \log n$ rounds. The preprocessing time and the space complexity will be polynomial in the number of nodes, which is $\beta^{O(1)} = n^{O(c)}$.

\subsubsection{Deterministic expander routing} Chang and Saranurak~\cite{ChangS20} showed that expander routing can be solved in $L \cdot \poly(1/\phi) \cdot 2^{O\left(\log^{2/3}n \log^{1/3} \log n\right)}$ rounds deterministically. The bound was recently improved to $L \cdot \poly(1/\phi) \cdot 2^{O\left(\sqrt{\log n \log \log n}\right)}$ by Chang, Huang, and Su~\cite{chang2024deterministic}.
It is difficult to calculate the precise space complexity of these algorithms, as they are inherently complicated. These algorithms are not space-efficient because there must be a sub-polynomial and super-polylogarithmic factor in the space complexity $\mu$. These algorithms have a similar recursive structure to the Ghaffari--Kuhn--Su algorithm. The main difference is that here in~\cite{chang2024deterministic,ChangS20}, some deterministic algorithms are used in the computation of a simultaneous embedding of an expander to each part of the partition $V_1 \cup V_2 \cup \cdots \cup V_{\beta}$. The congestion $c$ and the dilation $d$ of the embedding are sub-polynomial and super-polylogarithmic.  Existing techniques in deterministic algorithms do not allow us to make $c$ and $d$ polylogarithmic.

\subsection{Round-space tradeoff for expander routing}\label{sect:routing-tradeoff}

For simplicity of presentation, in the subsequent discussion, we sometimes assume $\tmix = \poly \log n$.
From the discussion in \cref{sect:routing-algo}, we know that when $L \geq n^c$ for some constant $c > 0$, then expander routing can be solved in $L \cdot \poly \log n$ rounds with memory cost $\mu = \deg(v) \cdot O(L)$ per node $v\in V$. That is, we can simultaneously achieve the best-known round complexity with the ideal memory cost.  We summarize this observation as a lemma.

\begin{lemma}\label{lem:tradeoff-1}
In a graph $G=(V,E)$ with $\tmix = \poly\log n$, if $L$ is at least polynomial in $n$, then expander routing can be solved with $\mu = \deg(v) \cdot O(L)$  space per node $v \in V$ in $T = L \cdot \poly \log n$ rounds with high probability in \switch.
\end{lemma}

Things are more complicated when $L = n^{o(1)}$, in which case there is a tradeoff between the round complexity and the space complexity depending on which routing algorithm is used. When $L = n^{o(1)}$, for all algorithms discussed in  \cref{sect:routing-algo}, the memory cost is always super-linear in $\deg(v)$ for each node $v \in V$. In particular, we may not set $\mu=O(n)$ even when $L=O(1)$.

To deal with the above issue, we demonstrate a method to obtain a smooth round-space tradeoff for any given expander routing algorithm. In the subsequent discussion, we only focus on the case of $L=O(\log n)$, since we already know that the case where $L$ is large is easy, and also observe that we can always reduce the parameter $L$ to $O(\log n)$ by processing the routing requests in batches.

We will show that for any given constant $\alpha > 0$, we can reduce the space complexity $\mu$ of the routing algorithm by a factor of roughly $\alpha$  at the cost of increasing the round complexity of the routing algorithm  by a factor of roughly $\alpha^2$. 
Here we only consider the space complexity of the routing algorithm itself and do not consider the memory cost of generating and storing the messages at the sources and the destinations.
Therefore, even though the round-space tradeoff allows us to reduce the space complexity to $\deg(v) \cdot o(L)$, such a space complexity may or may not be feasible, depending on the application.

\subsubsection{The algorithm for the round-space tradeoff}
The algorithm for the round-space tradeoff is as follows. Like the existing routing algorithms, we let each node $v \in V$ locally simulate $\deg(v)$ virtual nodes, so we have $2|E|$ total number of virtual nodes. 
As discussed earlier, we can translate each routing request in the original graph into a routing request between two virtual nodes such that the maximum load per node is still $O(\log n)$.

Each virtual node samples itself with probability $1/\alpha$. We then emulate a random graph of degree $O(\log n)$ with these sampled virtual nodes only. We only aim for delivering messages whose source and destination are both sampled virtual nodes.
This random graph can be constructed by letting each sampled virtual node initiate  $O(\alpha \log n)$ lazy random walks of length $2\tmix$. With high probability, $O(\log n)$ of them will end at a sampled virtual node. In the actual implementation, we simulate these random walks in $\alpha$ batches of size $O(\log n)$, in parallel for all sampled virtual nodes. At any time step of the parallel random walk simulation, the load of a node $v$ is at most $\lceil\deg(v)/\alpha \rceil \cdot O(\log n)$ with high probability, so each round can be simulated with $O(\log n)$ rounds in the \congest model. The round complexity of the random graph computation is $O(\alpha \tmix \log n)$. The space complexity of the random graph computation is $\lceil\deg(v)/\alpha \rceil \cdot O(\tmix \log n)$. We omit the details as they are similar to the proofs in~\cite{GhaffariKS17}.

After constructing the random graph, we run any existing expander routing algorithm on this random graph. In the routing task on the random graph, the maximum load per node can still be upper bounded by $O(\log n)$ with high probability.  The random graph has mixing time $O(\log n)$ with high probability. The embedding of the random graph has congestion $c = O(\tmix \log n)$ and dilation $d = 2\tmix$, so each round in the random graph can be simulated in $O(\tmix \log n)$ rounds in the original graph.

There is one subtle issue in implementing the above approach: Each sampled virtual node needs to know whether the destinations of its messages are also sampled virtual nodes. Same as~\cite{GhaffariKS17}, to deal with this issue, instead of doing the sampling with full randomness, we do it with $O(\log n)$-wise independent randomness. This is enough since we only need to show that certain bad events occur with probability at most $1/\poly(n)$. Storing these $O(\log n)$-wise independent random variables costs only an additive $\poly \log n$ term in the space complexity.

We will repeat the above procedure for  $O(\alpha^2 \log n)$ iterations. In each iteration, for each message, with probability $1/\alpha^2$, both its source and destination are sampled, and the message is successfully sent in that iteration. Therefore, after $O(\alpha^2 \log n)$ iterations, we have taken care of all messages. 

Roughly speaking, we save a factor of roughly $\alpha$ in the space complexity due to the fact that in each iteration, each node $v \in V$ in the original graph only has $\deg(v) / \alpha$ sampled virtual nodes in expectation. The increase in the round complexity is a factor of roughly $\alpha^2$ because we repeat for $O(\alpha^2 \log n)$ iterations.

\subsubsection{The round and space complexities}
We analyze the round complexity and space complexity of the algorithm. We write $T(n,\tmix)$ and  $S(n,\tmix) \cdot \deg(v)$ to denote the round complexity and space complexity of the given expander routing algorithm with $L=O(\log n)$. The round complexity and space complexity of the algorithm of the above round-space tradeoff with parameter $\alpha$ are as follows. 

The round complexity is 
$$O(\alpha^2 \log n) \cdot \left( O(\tmix \log n) \cdot \left( T\left(\poly(n), O(\log n)\right) +  O(\alpha \tmix \log n)\right)\right),$$
where $T\left(\poly(n), O(\log n)\right)$ is the cost of the routing algorithm in the random graph, $O(\alpha \tmix \log n)$ is the cost of constructing the random graph, $O(\tmix \log n)$ is the overhead in simulating one round of the random graph in the original graph, and $O(\alpha^2 \log n)$ is the number of iterations.

The space complexity is 
\begin{align*}
 &S\left(\poly(n), O(\log n)\right) \cdot 
 O(\log n) \cdot O\left(\lceil\deg(v)/\alpha \rceil + O(\log n)\right)\\
 &+  \lceil\deg(v)/\alpha \rceil \cdot O(\tmix \log n)\\
 &+ \poly\log n.  
\end{align*}
Here $S \left(\poly(n), O(\log n)\right) \cdot O(\log n)$ is the space complexity  of the routing algorithm in the random graph, as the graph has degree $O(\log n)$. To translate the space complexity into the space complexity of the original graph, we multiply this by $O\left(\lceil\deg(v)/\alpha \rceil + O(\log n)\right)$, which is the number of sampled virtual nodes for $v$. The term $\lceil\deg(v)/\alpha \rceil \cdot O(\tmix \log n)$ corresponds to the space complexity of the construction of the random graph, and the term $\poly\log n$ is the cost of storing $O(\log n)$-wise independent random numbers.

Let us do the calculation for the Ghaffari--Li algorithm~\cite{GhaffariL2018}. Suppose $\tmix=\poly \log n$ and $L = O(\log n)$. The above algorithm with parameter $\alpha$ has space complexity $$\lceil \deg(v)/\alpha \rceil \cdot  2^{O\left(\sqrt{\log n}\right)}$$ and round complexity $$\alpha^2 \cdot 2^{O\left(\sqrt{\log n}\right)}.$$

For example, with $\alpha = 2^{O\left(\sqrt{\log n}\right)}$, we may achieve the memory bound $\mu = O(n)$ with the same round complexity bound $2^{O\left(\sqrt{\log n}\right)}$ in the original Ghaffari--Li algorithm~\cite{GhaffariL2018}. We can afford to make $\mu$ sublinear in $n$. For example, with $\alpha = n^{0.1+o(1)}$, we can achieve $\mu = O(n^{0.9})$ with the round complexity $n^{0.2 + o(1)}$. We summarize the above discussion as a lemma.

\begin{lemma}\label{lem:tradeoff-2}
For any $\alpha \geq 1$, in a graph $G=(V,E)$ with $\tmix = \poly\log n$, expander routing can be solved with $\mu = \lceil \deg(v)/\alpha \rceil \cdot  2^{O\left(\sqrt{\log n}\right)}$ space per node $v \in V$ in $T = L \cdot \alpha^2 \cdot 2^{O\left(\sqrt{\log n}\right)}$ rounds with high probability in \switch.
\end{lemma}

As discussed earlier, we emphasize the space complexity of \cref{lem:tradeoff-2} does not consider the cost of preparing the messages at the source nodes and the cost of storing the messages at the destination nodes. In the worst case, such a cost can be as high as $\deg(v) \cdot L$ for a node $v \in V$.

Observe that \cref{lem:tradeoff-2} achieves sublinear-in-degree memory in a way that for each node $v \in V$, it is guaranteed that in each round, $v$ never receive more than $\mu = \lceil \deg(v)/\alpha \rceil \cdot  2^{O\left(\sqrt{\log n}\right)}$ messages. Therefore, the sublinear memory bound of \cref{lem:tradeoff-2} is not sensitive to the choice of the model for the ordering of the incoming messages within a round in the regime of $\mu < \deg(v)$.

\subsection{Existing expander decomposition algorithms}\label{sect:decomposition}

In the \congest model, expander routing is often used in combination with expander decomposition, so in this section, we examine the space complexity of existing distributed construction of expander decomposition~\cite{ChangPSZ21,ChangPZ18,changS19,ChangS20}.
Roughly, there are two different approaches in distributed expander decomposition algorithms: the nibble algorithm of Spielman and Teng~\cite{spielman2004nearly} and cut-matching games~\cite{khandekar2007cut,KRV09,RST14}. As we will later see, the distributed algorithms~\cite{ChangPSZ21,ChangPZ18,changS19} based on the nibble algorithm admits space-efficient implementation, and the distributed algorithms~\cite{ChangS20} based on cut-matching games seem to inherently require a large memory per node.

\subsubsection{Distributed expander decompositions via the nibble algorithm}\label{sect:decomposition-nibble}
It was shown in~\cite{ChangPSZ21,changS19} that a strong $(\epsilon,\phi)$-expander decomposition with $\phi = 1/\poly(1/\epsilon, \log n)$ can be computed in $O(n^c)$ rounds with high probability, where $c > 0$ can be any small constant. Recall that in a strong  $(\epsilon,\phi)$-expander decomposition, it is required that each $G\{V_i\}$ is a $\phi$-expander.
We briefly review their algorithm and check that the algorithm costs  $\deg(v) \cdot O(\log n)$ space per node $v \in V$.

At a high level, their algorithm recursively finds a low-diameter decomposition or a sparse cut until each remaining connected component is a $\phi$-expander.
A low-diameter decomposition removes a small fraction of edges to decompose the graph into small-diameter clusters. There are two popular ways to compute such a decomposition in polylogarithmic rounds. One is the randomized algorithm of Miller, Peng, and Xu~\cite{miller2013parallel}. The other one is the deterministic algorithm of Rozho\v{n} and Ghaffari~\cite{RozhonG20}.

The Miller--Peng--Xu algorithm is very simple: Each node grows its cluster by a BFS with a random start time, and then once the BFS reaches a node that has not joined any cluster, the node joins the cluster. The algorithm can be implemented with space complexity $\mu = O(1)$ per node, as each node just needs to remember its own random start time and the identifier of the cluster that it has joined. If multiple clusters reach a node simultaneously in a round, then a node will receive multiple messages in that round. Even in this case, the node does not need a super-constant space to process these messages, as the tie can be broken arbitrarily. 
The Rozho\v{n}--Ghaffari algorithm is inherently less space-efficient, as the algorithm requires storing a simultaneous embedding of trees where each edge $e$ can be involved in the embedding of $O(\log n)$ trees, so its space complexity is at least $\deg(v) \cdot O(\log n)$ per node $v \in V$.

Both two low-diameter decomposition algorithms can be used in the expander decomposition algorithm of~\cite{ChangPSZ21,changS19}. 
The Rozho\v{n}--Ghaffari algorithm was used in~\cite{ChangPSZ21}, and a modified version of the Miller--Peng--Xu algorithm was used in~\cite{changS19}. The purpose of the modification in~\cite{changS19} is to ensure that the bound on the number of inter-cluster edges holds not only in expectation but also with high probability. In the modified version, there is a preprocessing step that requires each node to estimate the size of its $k$-distance neighborhood for some $k$ up to some constant factor. After that, a partial cluster is computed. The computation of the partial cluster is simple and requires only BFS of a bounded number of steps. It is not hard to verify that the preprocessing step can still be done with $\mu = O(\log n)$ space per node.
To minimize the space complexity, we always use the Miller--Peng--Xu algorithm and not the Rozho\v{n}--Ghaffari algorithm for the low-diameter decomposition computation, so the memory cost is $\mu = O(\log n)$ space per node for this part of the expander decomposition algorithm.

Next, we consider the sparse cut computation in~\cite{ChangPSZ21,changS19}, which is based on the nibble algorithm of Spielman and Teng~\cite{spielman2004nearly}.
The nibble algorithm works roughly as follows. For some integer $t$, some node $v$, and some threshold $\epsilon'$, we calculate the probability distribution of a lazy random walk of length $t$ starting from $v$ in such a way that whenever the normalized probability at a node is smaller then $\epsilon'$, we truncate the probability to zero.
Here the normalized probability is defined as the random walk probability divided by the degree of the node. After finishing the calculation, we rank all the nodes in the non-increasing order of normalized probabilities.
It was shown in~\cite{spielman2004nearly} that if the underlying graph has a sparse cut, then with suitable choices of $t$, $v$, and $\epsilon'$, there exists an index $j$ such that taking the set of the first $j$ nodes in the above ranking gives us a sparse cut.

The sparse cut computation in~\cite{ChangPSZ21,changS19}  runs multiple instances of the nibble algorithm in parallel, with a guarantee that each node can participate in at most $O(\log n)$ instances of nibble algorithms in the sense that each node has non-zero random walk probability after truncation  for at most $O(\log n)$ instances of the nibble algorithms, with high probability.

After the random walk probability calculation,  the sparse cut computation is done in parallel, over all instances of the nibble algorithm. For each instance, the space complexity needed for the sparse cut computation is only $O(1)$ per node, as all we need to do is to find a tree spanning all nodes with non-zero random walk probabilities, and then these nodes can communicate along the tree edges. Although a node might receive a lot of messages, it does not need a large memory to process and store the messages, as all these messages can be combined together. For example, if we want to compute the volume of a subset of nodes, then we just need to compute the summation of the degrees over all nodes in the subset and do not need to remember the degree of each node in the subset. Since each node participates in at most $O(\log n)$ instances of the nibble algorithm, the overall space complexity needed is $O(\log n)$ per node for the sparse cut computation.

The bottleneck in the space complexity is the random walk probability computation. Although each node indeed only needs to maintain at most $O(\log n)$ truncated random walk probabilities at any time, a node $v \in V$ may still receive up to $\deg(v) \cdot O(\log n)$ random walk probabilities from its neighbors, so $v$ still needs to allocate a space of size  $\deg(v) \cdot O(\log n)$ to store and process these incoming messages. For this reason, the distributed expander decomposition algorithm of~\cite{ChangPSZ21,changS19} costs $\deg(v) \cdot O(\log n)$ space per node $v \in V$.

\subsubsection{Distributed expander decompositions via cut-matching games}\label{sect:decomposition-cutmatch}

A weak $(\epsilon,\phi)$-expander decomposition removes at most $\epsilon$ fraction of edges such that each remaining connected component is a $\phi$-expander. Using cut-matching games~\cite{khandekar2007cut,KRV09,RST14},  it was shown in~\cite{ChangS20} that for $\phi = 1/\poly(1/\epsilon, \log n)$, a weak $(\epsilon, \phi)$-expander decomposition can be computed in $\poly(1/\epsilon, \log n)$ rounds with high probability, and for $\phi = \poly(\epsilon)\cdot n^{-o(1)}$, a weak $(\epsilon, \phi)$-expander decomposition can be computed in $\poly(1/\epsilon)\cdot n^{o(1)}$ rounds deterministically.

Roughly speaking, a cut-matching game is a procedure that either computes an embedding of an expander to the underlying graph with small congestion and dilation or computes a sparse cut. Due to the need to store the embedding, this approach naturally requires super-linear-in-degree memory, just like the routing algorithms discussed earlier.

Since the algorithms of~\cite{ChangS20} do not require any local computation that is space-inefficient, the trivial space complexity bound $\mu = O(\deg(v)) \cdot T(n,\epsilon)$ holds, where $T(n,\epsilon)$ is the round complexity of the expander decomposition algorithm. For the randomized one,  $T(n,\epsilon) = \poly(1/\epsilon, \log n)$. For the deterministic one,  $T(n,\epsilon) = \poly(1/\epsilon)\cdot n^{o(1)}$. By very carefully examining the details of the algorithms, it might be possible to improve the trivial memory bound a bit, but it is unlikely that we can achieve $\mu = O(\deg(v))$ or better unless we make a significant change to the algorithm.

\subsection{Round-space tradeoff for expander decomposition}\label{sect:decomposition-tradeoff}
In this section, we show that the expander decomposition 
 algorithm of~\cite{ChangPSZ21,changS19}, which is based on the nibble algorithm of Spielman and Teng~\cite{spielman2004nearly}, can be implemented in such a way that yields a smooth round-space tradeoff. As discussed earlier, the distributed algorithm of~\cite{ChangPSZ21,changS19} inherently does not require the nodes to store a lot of information, 
but each node still needs a large space to process the incoming messages, and that will be the main difficulty in designing a space-efficient implementation.  

  Our goal is to show that for any $\alpha \geq 1$, the expander decomposition 
 algorithm of~\cite{ChangPSZ21,changS19} can be implemented to cost only $\mu = \lceil \deg(v) / \alpha \rceil \cdot O(\log n)$ space per node $v \in V$ and takes in $\alpha \cdot O(n^c)$ rounds with high probability, where $c > 0$ can be any small constant. That is, roughly speaking, the space complexity can be improved by a  factor of $\alpha$ at the cost of increasing the round complexity by a factor of $\alpha$. 
 
If each node $v$ has a memory of size $O(\deg(v))$, then $v$ can afford to remember all its deleted incident edges.
For the case $\mu = o(\deg(v))$, this is not doable.
In this case, instead of remembering all the deleted edges, we may simply let the nodes in the same connected component share a common identifier to indicate that they belong to the same component, and this costs only $O(1)$ space per node. When a node sends a message, it always attaches the identifier of its component. If a node $u$ receives a message from a node $v$ such that $u$ and $v$ belong to different components, then $u$ simply discards the message. 

We already know that the low-diameter decomposition needed in~\cite{changS19} can be computed with $O(\log n)$ space per node, so in the subsequent discussion, we focus on the implementation of the nibble algorithm.  As discussed earlier, after calculating the probabilities, the sparse cut computation costs only $O(\log n)$ space per node, as each node is involved in at most $O(\log n)$ instances of the nibble algorithm, so we only need to consider the calculation of the  random walk probabilities.

In the calculation of the  random walk probabilities, a node $v \in V$ may receive up to $\deg(v) \cdot O(\log n)$ random walk probabilities from its neighbors, so $v$ still needs to allocate a space of size  $\deg(v) \cdot O(\log n)$ to store and process these numbers in a straightforward implementation. To improve the above space complexity, we  allocate $\alpha$ rounds to do the above task. We let each instance of the nibble algorithm generate a random number  $x \in \{1,2, \ldots, \alpha\}$, and all the random walk probabilities associated with this instance will be transmitted at the $x$th round. Therefore, in each round of communication, with high probability, each node $v$ receives random walk probabilities from at most $\lceil \deg(v)/\alpha \rceil \cdot O(\log n)$ instances of the nibble algorithm, so $v$ only needs to allocate a space of size $\lceil \deg(v)/\alpha \rceil \cdot O(\log n)$ to store the summation of probabilities. After finishing the calculation of the summations, we may apply the truncation, and only $O(\log n)$ of them can be nonzero, with high probability, so the space of size $\lceil \deg(v)/\alpha \rceil \cdot O(\log n)$ can be reused in the subsequent rounds.
To summarize, at the cost of increasing the round complexity by a factor of $\alpha$, we may reduce the space complexity to $\lceil \deg(v) / \alpha \rceil \cdot O(\log n)$.

\begin{lemma}\label{lem:decomposition}
For any $\alpha \geq 1$ and $0 < \epsilon < 1/2$, 
a strong $(\epsilon,\phi)$-expander decomposition of  $G=(V,E)$ with $\phi = 1/\poly(1/\epsilon, \log n)$ can be computed with $\mu = \lceil \deg(v) / \alpha \rceil \cdot O(\log n)$ space per node $v \in V$ in $\alpha \cdot O(n^c)$ rounds with high probability in \switch, where $c > 0$ can be any small constant. 
\end{lemma}

Unlike \cref{lem:tradeoff-2}, here in \cref{lem:decomposition}, when $\mu = o(\deg(v))$, we allow the possibility that a node $v$ sends and receives more than $\mu$ messages in one round. We emphasize that the correctness of the lemma is not sensitive to the choice of the model of how the incoming messages arrive within a round. In particular, the lemma still holds even if the ordering of the incoming messages in a round is \emph{adversarial} and $v$ can only process the messages in a streaming fashion with a space bound $\mu$.

If we want to improve \cref{lem:decomposition} to ensure  that each node never receives more than $\mu$ messages in a round even when $\mu = o(\deg(v))$, one way to achieve this is to pay an extra factor of $\alpha$ in the round complexity: Using $\alpha$ rounds to simulate one round of the underlying distributed algorithm and letting each node send its messages at a random round in these $\alpha$ rounds. With high probability, a node $v$ receives at most $\mu = \lceil \deg(v) / \alpha \rceil \cdot O(\log n)$ messages in a round.

\section{Triangle Listing in \switch} \label{appendix:triangle-listing}
\subsection{Tools for Clique Listing in \switch}

We begin by presenting a set of basic building blocks for clique listing.

Denote by $N(v)$ the neighbors of node $v$. We begin with the following straightforward claim.

\begin{theorem}\label{thm: ub-listing-congest-local-listing}
Any node $v$ can learn all of the $K_3$ which it is a part of within $O(\deg(v))$ rounds of the \switch model. Further, this uses only the edges incident to $v$ for communication.
\end{theorem}

\begin{proof}[Proof of \Cref{thm: ub-listing-congest-local-listing}]
Node $v$ iterates over its $\deg(v)$ neighbors, each time choosing some $u \in N(v)$. It broadcasts $u$ to $N(v)$ in $1$ round, and each $w \in N(v)$ responds whether or not the edge $\{u, w\}$ exists. Node $v$ now lists any triangle it sees, and continues to iterate over $N(v)$.
\end{proof}

We now slightly extend a claim presented in \cite{CensorCLL21} that says that if the average degree in a graph is small, then we can use the above algorithm to remove at least half of the nodes in the graph.
\begin{lemma}[Claim 4.1 in \cite{CensorCLL21}]\label{ub-listing-congest-helper1}
Let $c > 1$ be any value. Denote by $\gamma = 2m/n$ the average degree of a graph. If we remove all nodes with degree at most $c\gamma$, then we are left with less than $n/c$ nodes.
\end{lemma}
\begin{proof}[Proof of \Cref{ub-listing-congest-helper1}]
    Assume for the sake of contradiction that we are left with at least $n/c$ nodes -- i.e. there are at least $n/c$ nodes with degree more than $c\gamma$. This means that there are more than $(n/c)\cdot(c\gamma/2) = n\gamma/2$ edges in the graph, i.e. $m > n\gamma/2$. However, by definition, $\mu = n\gamma/2$, which is a contradiction.
\end{proof}

We also slightly extend a claim from \cite{CensorCLL21} that says the following. Assume that we have some partitioning of the graph $G=(V, E)$ such that $V = V_1 \cup \dots \cup V_x$ for some $x$ and $V_i \cap V_j = \emptyset$ if $i\neq j$. For each $V_i$, denote $E_i = E \cap (V_i \times V_i)$, and let $C_i = (V_i, E_i)$ be called a \emph{cluster}. Further, denote by $\gamma = 2m/n$ the average degree of the graph. Let there be any value $c > 1$. Finally, call a cluster $C_i$ a \emph{low-average-degree cluster} if $2|E_i|/|V_i| < \gamma / c$ and denote by $V_{low}, E_{low}$ the set of all nodes, edges (respectively) in low-average-degree clusters. Then, the following holds.

\begin{lemma}[Claim 4.4 in \cite{CensorCLL21}]\label{ub-listing-congest-helper2}
Let $d > 1$ be some value and assume the definitions from the paragraph above. If the total number of edges inside clusters is at least a $1/d$ fraction of the total edges in the graph (i.e. $\sum_{i \in [x]}|E_i| \geq m/d$), then $|E_{low}| \leq d\cdot m/c$.
\end{lemma}
\begin{proof}[Proof of \Cref{ub-listing-congest-helper2}]
    Assume for the sake of contradiction that $|E_{low}| > d\cdot m/c$. As for every low-average-degree cluster $C_i$ it holds that $2|E_i|/|V_i|<\gamma/c$, then it also holds that $|E_i| < |V_i|\cdot \gamma/ (2c)$ and so from linearity we get that $|E_{low}| < |V_{low}|\cdot \gamma/ (2c)$. Put together, this emplies that $d\cdot m/c < |E_{low}| < |V_{low}|\cdot \gamma/(2c)$. Plugging in $2m/n = \gamma$, we get that $d\cdot m < |V_{low}|\cdot (2m/n)/2$, which means $d\cdot n < |V_{low}|$. As $d>1$, this implies that $n < |V_{low}|$, which is a contradiction.
\end{proof}

We prove the following auxiliary routines which are specific to the \switch model. Denote by $D$ the diameter of the given graph.

\begin{lemma}\label{ub-listing-congest-helper3}
Let $x$ be some value such that $0 \leq x \leq \mu$ and each node $v$ originally holds $x$ values $v_1, \dots, v_x$. Further, let $\circ$ be some aggregation function (i.e. commutative and associative). It is possible to construct a directed tree $T$ where: (1) $T$ spans $V$, (2) $T$ is a subgraph of $G$, (3) the diameter of $T$ is $O(D)$, and (4) for every node $v$ denote by $Des(v)$ the descendants of $v$ in $T$; for each $i \in [x]$, node $v$ knows the values $v^{tot}_i=\sum_{u \in Des(v) \cup \{v\}}u_i$, where $\sum$ applies $\circ$. This algorithm requires $O(x + D)$ rounds and each node holds at most $x \geq \mu$ additional messages during the algorithm compared to what it holds before the algorithm begins.
\end{lemma}
\begin{proof}[Proof of \Cref{ub-listing-congest-helper3}]
Let $v_0$ be some arbitrary node, construct a BFS tree from $v_0$ and denote this tree by $T$. Clearly, $T$ spans $V$, is a subgraph of $G$, and has diameter of at most $2D = O(D)$. Using pipelining, we perform the $\circ$ operation up the tree in $O(x + D)$ rounds. At any given time, each node stores at most $x \leq \mu$ new values (on top of what it stored before the execution of the algorithm) -- specifically, the $x$ intermediate results of applying $\circ$ on the values received from its children in the tree.
\end{proof}

\begin{lemma}[Lemma 1 in \cite{ChangPSZ21}]\label{ub-listing-congest-helper4}
In $O(\log(n) + D)$ rounds we can assign each node a new identifier in $[|V|]$ with the following property: this algorithm requires $O(\log n)$ memory per node throughout its computation, and if this memory is stored after the computation, then every node $u$ can locally compute for every node $v$ the value $\left \lfloor \log(\deg(v)) \right \rfloor$.
\end{lemma}
\begin{proof}[Proof of \Cref{ub-listing-congest-helper4}]
First, every node $v$ creates $\log (n)$ boolean values: $v_1, \dots, v_{\log(n)}$, where $v_i = 1$ if and only if $\left \lfloor \log(\deg(v)) \right \rfloor = i$. Then, using \Cref{ub-listing-congest-helper3}, we generate a tree $T$ in $O(\log(n) + D)$ rounds such that every node knows the point-wise sum of the boolean vectors of its descendants in $T$. Let $r$ be the root of $T$. Notice that $r$ now holds the frequencies of the log degrees of all the nodes in the graph -- that is, values $d_1, \dots, d_{\log (n)}$ such that $d_i$ is how many nodes $v$ exist such that $\left \lfloor \log(\deg(v)) \right \rfloor = i$. Node $r$ breaks $[|V|]$ into intervals matching $d_1, \dots, d_{\log (n)}$, i.e. the first interval is $[1, d_1]$, followed by $[d_1 + 1, d_1 + d_2]$, and so forth. It takes the first identifier in the interval matching its degree, i.e. the interval for $\left \lfloor \log(\deg(r)) \right \rfloor$ and forwards the remaining intervals to its children. Each node in the tree repeats this process until every node has a new identifier. This takes $O(\log (n) + D)$ rounds using pipelining. Finally, in a final $O(\log (n) + D)$ rounds, node $r$ broadcasts $d_1, \dots, d_{\log(n)}$ to all the graph -- given this information, every node $u$ can compute for every node $v$ the value $\left \lfloor \log(\deg(v)) \right \rfloor$, as it knows to which interval $v$ belongs.
\end{proof}

\subsection{General Routing Tool in \switch}

We desire to show the following general routing theorem. This subroutine proves very useful in our applications to clique listing in \switch.

\begin{theorem}\label{congest-general-routing1}
    Let $G$ be a graph with $\tmix = \poly\log n$, $\Delta$ its maximal degree, $\mu$ the memory bound, $1 \leq x \leq \Delta$ be some positive integer and $S \subseteq V$ be some nodes where the degree (in $G$) of every node in $S$ is at least $x$. Assume that all nodes in the graph know which nodes belong to $S$. Let there be a routing instance where nodes $V$ desire to send messages to nodes $S$, and the following hold:
    \begin{enumerate}
        \item Every node in $S$ is the target of at most $\mu$ messages.
        \item Each node $v \in V$ originally has at most $\deg(v)$ values $m_{v, 1}, \dots, m_{v, \deg(v)}$ which need to be sent to nodes in $S$. Each value might potentially be sent as multiple messages to different nodes in $S$ -- for value $m_{v, i}$, denote by $A_{v, i} \subseteq S$ the set of nodes in $S$ which need to receive $m_{v, i}$.
        \item Given $m_{v, i}$, all nodes can locally compute $A_{v, i}$.
    \end{enumerate}
    Then, this routing can be accomplished in $\mu/x \cdot n^{o(1)}$ rounds in \switch.
\end{theorem}
\begin{proof}[Proof of \Cref{congest-general-routing1}]
The idea behind this subroutine is to first shift all the values from $V$ to $S$ in a way that every node $v \in S$ holds some $\mu$ messages (not necessarily the ones intended towards $v$), and then reshuffle the messages with $S$ using \Cref{lem:tradeoff-2} -- this allows us to use \Cref{lem:tradeoff-2} with a lower bound of $x$ on the degree of every node sending or receiving a message, as all nodes in $S$ have degree at least $x$.

Every node $v \in V$ locally computes the value $cost(v) = |A_{v, 1}|+\dots+|A_{v, \deg(v)}|$. Notice that if we sum the costs of all nodes, we receive at most $|S| \cdot \mu$ as each node in $S$ is the target of at most $\mu$ messages, and so no more than $|S| \cdot \mu$ messages are desired to be sent. We now use \Cref{ub-listing-congest-helper3} to compute a tree $T$ such that every node in the tree knows the sum of the costs of all the nodes in its subtree; this takes $\widetilde O(1)$ rounds as $\tmix = \poly\log n$. Using $T$, we split the interval $[1, |S| \cdot \mu]$ into non-overlapping subintervals such that each node $v$ has a contigous interval of size $cost(v)$ -- i.e. the root $r$ of the tree takes the interval $[1, cost(r)]$, and the splits $[cost(r), |S| \cdot \mu]$ across its children according to the total cost of each of their subtrees; in turn, this propagates down the tree.

Now, we construct a routing instance between $V$ and $S$ and use \Cref{lem:tradeoff-2} to execute it. We subdivide $[1, |S| \cdot \mu]$ into $S$ intervals ($[1, \mu]$, $[\mu + 1, 2\mu]$, etc.) such that each node in $S$ is responsible for one interval. Recall that we assume that every node in $V$ knows which nodes are in $S$, and so every node locally knows which node in $S$ is responsible for which interval. Next, every node $v \in V$ takes its subinterval of size $cost(v)$ and splits it into subintervals of sizes $|A_{v, 1}|, \dots, |A_{v, \deg(v)}|$, according to each of its values $m_{v, 1}, \dots, m_{v, \deg(v)}$. For a given value $m_{v, i}$, denote by $I(m_{v, i})$ the subinterval assigned to $m_{v, i}$ and by $S(m_{v, i}) \subseteq S$ the nodes in $S$ whose intervals overlap $I(m_{v, i})$. Finally, for every $m_{v, i}$, node $v$ desires to send the messages $m_{v, i}$ and the endpoints of $I(m_{v, i})$ to the first node in $S(m_{v, i})$. Notice that in this routing task, every $v \in V$ desires to send at most $3 \cdot \deg(v)$ messages, and each node in $S$ receives at most $3\mu$ messages. Thus, this routing task can be accomplished in $\mu/x \cdot n^{o(1)}$ rounds using \Cref{lem:tradeoff-2}.

Observe that now for each $m_{v, i}$, the first node in $S(m_{v, i})$ knows $m_{v, i}$ as well as $I(m_{v, i})$. Let $s \in S(m_{v, i})$ be this node. Using $I(m_{v, i})$, node $s$ can compute $S(m_{v, i})$. Thus, we now wish to duplicate the messages $m_{v, i}$ and $I(m_{v, i})$ such that all of $S(m_{v, i})$ know these -- i.e. node $s$ sends these messages to the second node in $S(m_{v, i})$, and then they respectively send the messages to the third and fourth nodes in $S(m_{v, i})$, and so forth. Each such iteration can be executed in $\mu/x \cdot n^{o(1)}$ rounds using \Cref{lem:tradeoff-2}, as each node in $S$ sends and receives at most $3\mu$ messages. Within $\log|S|$ iterations, every node in $S(m_{v, i})$ knows $m_{v, i}$ and $I(m_{v, i})$.

Finally, we are at the state that every node in $S(m_{v, i})$ knows $m_{v, i}$ and $I(m_{v, i})$ and recall that given each value $m_{v, i}$, any node in the graph can locally compute $|A_{v, i}|$. Thus, every node in $S(m_{v, i})$ can compute $|A_{v, i}|$ and knows to which of these nodes to forward the value $m_{v, i}$. As each node in $S$ desires to send and receive at most $\mu$ messages, this routing task is accomplished in $\mu/x \cdot n^{o(1)}$ rounds using \Cref{lem:tradeoff-2}.
\end{proof}

\subsection{Proof of \Cref{thm: ub-listing-congest}}

We are now ready to prove \Cref{thm: ub-listing-congest}.
To do so, we present a subroutine that removes some nodes and edges from the graph, while listing all the triangles these nodes or edges are involved in, such that either one of the following conditions hold: (1) a constant fraction of the nodes is removed, or, (2) a constant fraction of the edges is removed. This subroutine completes in $n^{1+o(1)}/\mu^{1/2}$ rounds, therefore, repeating this subroutine for $O(\log n)$ times concludes the proof of \Cref{thm: ub-listing-congest}. All that is left to do is to show this subroutine.

Denote by $\gamma = 2m/n$ the average degree of the graph -- this value exists but is not necessarily known to any node. The first action we perform in the subroutine is that for $x = \widetilde\Theta(1)$ (that is, a specific polylogarithmic value which is set to satisfy requirements further down in the proof), every node $v$ with degree $\deg(v) \leq x n^{1/3}$ lists all the $K_3$ which it is a part of, in $\tilde{O}(n^{1/3}) = \widetilde{O}(n/\mu^{1/2})$ rounds, using \Cref{thm: ub-listing-congest-local-listing}, as $x = \widetilde{O}(1)$. By using $c=2$ in \Cref{ub-listing-congest-helper1}, we get that if $2\gamma \leq xn^{1/3}$, then we removed at least half of the nodes in the graph at this stage.

Therefore, from here on we assume that $\gamma > (x/2)n^{1/3}$. We can assume this since our subroutine has to at least a constant fraction of the nodes or of the edges. If $\gamma \leq (x/2)n^{1/3}$, then we already removed at least half the nodes and so it does not matter what happens from here on. Therefore, we assume from here on that $\gamma > (x/2)n^{1/3}$. In the case that the assumption does not hold, then the following algorithm will not necessarily execute in $n^{1+o(1)}/\mu^{1/2}$ rounds, however, we can just stop it after $n^{1+o(1)}/\mu^{1/2}$ -- it is guaranteed that stopping it in the middle of operation does not hurt the correctness of the general algorithm, as no edge is removed from the graph without listing all of the triangles which it is a part of.

Now, execute the strong  $(\epsilon, \phi)$-expander decomposition  algorithm from \Cref{lem:decomposition} with $\epsilon = 1/\poly \log (n)$ and $\phi = 1/\poly \log (n)$. The algorithm costs $O(n^{0.001})$ rounds and uses only $O(\deg(v))=O(\mu)$ memory per node $v$ by setting $\alpha = \Theta(\log n)$ and $c = 0.0005$ in \Cref{lem:decomposition}. Refer to~\cref{sect:expander-appendix} for the precise definition of the decomposition computed by \Cref{lem:decomposition}. 

We summarize the key properties of the decomposition that we need.
Let $V = V_1 \cup V_2 \cup \cdots$ be the decomposition computed by \Cref{lem:decomposition}.
Let $E_i$ be the set of edges in the subgraph $G[V_i]$ induced by $V_i$.
Denote by $E_m = E_1 \cup E_2 \cup \cdots$ the edges of the clusters in the decomposition and by $E_r = E \setminus E_m$ the edges between the clusters. The following properties hold.

\begin{enumerate}
    \item $|E_r|\leq \epsilon|E| = |E|/\poly\log(n)$.
    \item The mixing time within each cluster is $O(\phi^{-2} \log n) =  \widetilde{O}(1)$.
    \item For a cluster $C = (V_i, E_i)$, and a node $v \in V_i$, it holds that $\deg_C(v) \geq \deg(v)/\poly\log(n)$, where $\deg_C(v)$ is the number of edges incident to $v$ in $C$. 
\end{enumerate}

Notice that  $|E_r|\leq |E|/\poly\log(n)$ implies $|E_m| \geq |E|/2$.
The third property follows from the fact that in a strong  $(\epsilon, \phi)$-expander decomposition, not only $G[V_i]$ but also $G\{V_i\}$ is an $\phi$-expander. For $G\{V_i\}$ to has conductance at least $1/\poly\log(n)$, it is necessary that $\deg_C(v) \geq \deg(v)/\poly\log(n)$ for all $v \in V_i$.


Using the definitions of \Cref{ub-listing-congest-helper2}, denote a cluster $C = (V_i, E_i)$ as a low-average-degree cluster if $2|E_i|/|V_i|<(x/2)n^{1/3}/8$. As we assume that $(x/2)n^{1/3} < \gamma$, this implies that if $C$ is a low-average-degree-cluster, then $2|E_i|/|V_i| < \gamma / 8$. Denote by $E_{low}$ all of the edges inside low-average-degree clusters. According to \Cref{ub-listing-congest-helper2}, as $|E_m| \geq |E|/2$, then $|E_{low}|\leq |E|/4$.

Denote by $E_{high} = E_m \setminus E_{low}$ all the edges in the clusters which are not low-average-degree clusters. From here on, the clusters whose edges are in $E_{high}$ are referred to as high-average-degree clusters. As $|E_m| \geq |E|/2$ and $|E_{low}| \leq |E|/4$, then $|E_{high}|\leq |E|/4$. We now proceed to listing all of the triangles with edges in $E_{high}$, which then allows us to safely remove all the edges in $E_{high}$. By doing so, we remove at least a constant fraction of the edges in the graph, as required.

Let $C = (V_i, E_i)$ be a high-average-degree cluster. We list all of the triangles involving at least one edge in $E_i$. Notice that these triangles can include other edges incident to $C$ from $V \setminus V_i$, but not edges which do not have incident nodes in $V_i$. As such, define $E_i^{'} = E \cap (V_i \times (V \setminus V_i))$ as the edges with one node in $V_i$ and the other outside it. Using this definition, we desire to list all triangles with edges in $E_i \cup E_i^{'}$. For simplicity, we denote by $V_i^{'} \subseteq V \setminus V_i$ the set of nodes outside of $V_i$ which have edges incident to them in $E_i^{'}$.

Compute $|V_i|$ using \Cref{ub-listing-congest-helper3} in $\widetilde{O}(1)$ rounds,
as the mixing time of $C$ is $\widetilde{O}(1)$ and so its diameter is $\widetilde{O}(1)$.
If $|V_i| \leq n^{1/3}$, then for any $v \in |V_i|$, we know that $\deg_C(v) \leq n^{1/3}$ and so according to the guarantees of the decomposition, $\deg(v) = \widetilde{O}(n^{1/3})$. As such, all nodes in $V_i$ can list all of the triangles they are involved in using \Cref{ub-listing-congest-helper1} and we are done.

Thus, we are now left with the case that $|V_i| > n^{1/3}$. Denote $n_i = |V_i|, m_i = |E_i|, m_i^{'} = |E_i^{'}|$ and also $\widetilde{m_i} = m_i + m_i^{'}$. Notice that due to the definition of the expander decomposition, it holds that $\widetilde{m_i} = \widetilde{O}(m_i)$. From here on, the term ``with high probability'' (``w.h.p.'') refers to $n_i$. However, as $n_i$ is a polynomial of $n$, then this implies high probability w.r.t. $n$ (by changing constants). Our resulting algorithm runs in $n^{1+o(1)}_i/\mu^{1/2} = n^{1+o(1)}/\mu^{1/2}$ rounds.

Using the notations from \cite{ChangPSZ21}, let $\delta = 2^{\left \lfloor \log(2m_i/n_i) \right \rfloor}$ be the average degree inside $C$ rounded down to the nearest power of 2. Notice that $\delta$ can be computed in $\widetilde {O}(1)$ rounds using \Cref{ub-listing-congest-helper3}. Denote $k_v = \deg_C(v)/\delta$. If $k_v < 1/2$, then $v$ is a \emph{class-0} node, and if $2^{i-2} \leq k_v < 2^{i-1}$, the $v$ is a \emph{class-$i$} node.
By executing \Cref{ub-listing-congest-helper4}, the nodes are relabelled such that every node can locally compute the class of every other node. Finally, notice that $\sum_{v \in V_i : k_v \geq 1/2} 2 k_v \geq n_i$.

At this moment, our algorithm diverges significantly from that of \cite{ChangPSZ21}. Recall that we observe the graph $C' = (V_i \cup V_i^{'}, E_i \cup E_i^{'})$ and desire to list all triangles in this graph. In the \congest algorithm, the nodes of $C'$ are uniformly at random split into $n^{1/3}_i$ sets $A_1, \dots, A_{n^{1/3}_i}$. As there are $n_i$ options to choose 3 sets, and also $\sum_{v \in V_i : k_v \geq 1/2} 2 k_v \geq n_i$, then if each node with $k_v \geq 1/2$ chooses between $2k_v$ and $4k_v$ triples of sets and learns the edges between them, then we will certainly list all triangles. Further, as the number of triples each node chooses is proportional to $k_v = \deg_C(v) / \delta$, then every node $v$ will receive a number of messages which is proportional to its degree, and this will guarantee the round complexity.

In our case of \switch, working with nodes from different \emph{classes} (i.e. nodes whose degrees differ greatly) is problematic as they will take different number of rounds to learn $\mu$ messages. Thus, we observe all the $\log(n_i)$ different node \emph{classes}. For each $j \in [|\log(n_i)|]$, denote by $K_j$ all nodes in $v$ which are \emph{class-$j$} nodes. Recall that $\sum_{v \in V_i : k_v \geq 1/2} 2 k_v \geq n_i$. Therefore, since $\sum_{v \in V_i : k_v \geq 1/2} 2 k_v = \sum_{j \in [\log(n_i)]}\sum_{v \in K_j} 2 k_v$, then there exists some $j^* \in [|\log(n_i)|]$ such that $\sum_{v \in K_{j^*}} 2 k_v \geq n_i/\log(n_i)$. From here on, we use only the nodes in $K_{j^*}$ for listing triangles

We must now break our algorithm into two cases: whether or not $\mu$ is large enough for a node to learn all of the edges between a given triplet of sets $A_i$. If $\mu$ is large enough, then we can proceed with an algorithm similar to that in the \congest model. Otherwise, we have to partition the nodes of $C'$ slightly differently.

In both cases we require the following lemma which is proven in \cite{ChangPSZ21}.

\begin{lemma}[2 in \cite{ChangPSZ21}]\label{ub-listing-congest-k3-helper}
Consider a graph with $\bar{n}$ nodes and $\bar{m}$ edges. We generate a subset $S$ of nodes, by letting each node join $S$ independently with probability $p$. Suppose that the maximum degree is $\Delta \leq \bar{m}p/20\log(\bar{n})$ and $p^2 \bar{m} \geq 400 \log^2(\bar{n})$. Then, w.h.p., the number of edges in the subgraph induced by $S$ is at most $6p^2\bar{m}$.
\end{lemma}

\paragraph{The case where $\widetilde{m_i} / n_i^{2/3} \leq \mu$.} In this case, we proceed with an algorithm similar to that in \congest. We create the sets $A_i$ above as follows.
Every node in $C^{'}$ uniformly at random chooses to join one of $A_1, \dots, A_{n^{1/3}_i}$ and informs its neighbors 
of its choice. Now, for every edge, both of its endpoints know to which set each of them belongs.

Observe any pair $A_{j}, A_{k}$ and denote $S = A_{j} \cup A_{k}$. Using the notations of \Cref{ub-listing-congest-k3-helper}, each node in $C^{'}$ chooses to join $S$ with probability $p = n^{-1/3}_i$ (or $p = 2\cdot n^{-1/3}_i$ if $j = k$). Let us verify that the conditions of \Cref{ub-listing-congest-k3-helper} hold. As  $C$ is a high-average-degree cluster, $2m_i/n_i > x \cdot n^{1/3}/16 \geq x \cdot n^{1/3}_i/16$ and as $x = \widetilde{\Theta(1)}$, then $m_i = \widetilde{\Omega}(n^{4/3}_i)$. Further, observe that $\Delta_{C^{'}} = \widetilde{O}(n_i)$ as every $v \in V_{i}$ has $\deg_{C}(v) \leq n_i$ as so $\deg(v) = \widetilde{O}(n_i)$, and every $u \in V_{i}^{'}$ has edges in $C^{'}$ only towards nodes in $V$ and so $\deg_{C^{'}}(u) \leq n_i$. As such, $\widetilde{m_i}\cdot p/20\log(n_i) = \widetilde\Omega(n^{4/3}_i \cdot n^{-1/3}_i) = \widetilde\Omega(n_i) \geq \Delta_{C^{'}}$, and so we finally get $\Delta_{C^{'}} \leq \widetilde{m_i}\cdot p/20\log(n_i)$. Further, $p^2\cdot \widetilde{m_i} = \widetilde\Omega(\widetilde{m_i}/n_i^{2/3}) = \widetilde\Omega(n_i^{2/3}) \geq 400 \log^2(n_i)$. Thus, as the conditions of the lemma hold, we know that the number of edges between $A_j$ and $A_k$ is at most $6p^2\cdot \widetilde{m_i} = O(\widetilde{m_i}/n_{i}^{2/3})$, w.h.p.

Each node in $K_{j^*}$ is uniquely assigned at most $2k_v \log(n_i)$ triples of sets $A_i$ -- as stated above, due to the choice of $K_{j^*}$, this is enough sets per node in $K_{j^*}$ to ensure that every possible choice of triplets $A_i$ is chosen by a node in $K_{j^*}$.
For each node $v \in K_{j^*}$, denote the set of triplets of sets $A_i$ it is assigned using $T(v)$. As all nodes can locally compute the \emph{class} of every node, then all nodes know which nodes are in $K_{j^*}$ and can locally compute $T(v)$ for any $v \in K_{j^*}$.

Finally, for every $v \in K_{j^*}$, node $v$ iteratively learns triplets from $T(v)$, lists the triangles it sees using these triplets, throws away this information and continues on to the next triplets. Notice that each triplet requires learning $y = O(\widetilde{m_i}/n_{i}^{2/3})$ edges. 

If $y \geq \deg_C(v)$, then node $v$ learns each triplet in $T(v)$ one at a time using \Cref{congest-general-routing1} with a memory bound $\mu' = \Theta(\widetilde{m_i}/n_{i}^{2/3}) = O(\mu)$ and $x = 2^{j^* - 2} \cdot m_i/n_i$. Notice that all the conditions of \Cref{congest-general-routing1} hold: $C$ has $\tmix = \poly \log n$, all nodes in $K_{j^*}$ have degree at least $2^{j^* - 2} \cdot m_i/n_i$, all the nodes in the graph know which nodes belong to $K_{j^*}$, node $v$ is the target of at most $\mu' = \Theta(\widetilde{m_i}/n_{i}^{2/3})$ messages, each node $u \in C$ has at most $\deg(u)$ values to send (i.e. its incident edges) and given each such edge, as all nodes can locally compute $T(v)$, then it is possible to know to which nodes in $K_{j^*}$ to send this edge at this moment. Thus, the routing completes in $\mu'/x \cdot n_i^{o(1)} = \widetilde{m_i}/n_{i}^{2/3} / (2^{j^* - 2} \cdot m_i/n_i) \cdot n_i^{o(1)}$ rounds. This is repeated for $|T(v)| \leq 2k_v \log (n_i)$ times, for a total complexity of $2k_v \log (n_i) \cdot \widetilde{m_i}/n_{i}^{2/3} / (2^{j^* - 2} \cdot m_i/n_i) \cdot n_i^{o(1)}$, and as $k_v \leq 2^{j^* - 1}$, this comes out to $\widetilde{m_i}/(n_{i}^{2/3} \cdot m_i/n_i)\cdot n_i^{o(1)} = n_i^{1/3} \cdot n_i^{o(1)} = n^{1/3 + o(1)}$ rounds, as required.

Otherwise, $y < \deg_C(v)$. In this case, $v$ may need to learn several triplets in $T(v)$ at once, in order to utilize all of its available bandwidth. As $2^{j^* - 2} \cdot m_i/n_i \leq \deg_C(v)$, then $v$ learns $\floor{(2^{j^* - 2} \cdot m_i/n_i)/y}$ triplets at once using \Cref{congest-general-routing1} with a memory bound $\mu' = 2^{j^* - 2} \cdot m_i/n_i \leq \deg_C(v) \leq \mu$ and $x = 2^{j^* - 2} \cdot m_i/n_i$, within $n^{o(1)}$ rounds. This is repeated $|T(v)|/\floor{(2^{j^* - 2} \cdot m_i/n_i)/y} \leq 2k_v \log(n_i) / \floor{(2^{j^* - 2} \cdot m_i/n_i)/y} = \widetilde{O}(y/(m_i/n_i)) = \widetilde {O}((\widetilde{m_i}/n_{i}^{2/3}) / (m_i/n_i)) = \widetilde{O}(n_i^{1/3})$ times, for a total of $n^{1/3+o(1)}$ rounds.

\paragraph{The case where $\widetilde{m_i} / n_i^{2/3} > \mu$.} In this case, we partition the nodes of $C'$ slightly differently. Set $s = (\widetilde{m}_i/\mu)^{1/2}$, and partition the nodes into sets $B_1, \dots B_s$, similarly to how the sets $A_i$ are chosen above.

Observe any pair $B_{j}, B_{k}$ and denote $S = B_{j} \cup B_{k}$. Using the notations of \Cref{ub-listing-congest-k3-helper}, $p = 1/s$ (or $p = 2/s$ if $j = k$). We verify the conditions of \Cref{ub-listing-congest-k3-helper}. As above, we know $\Delta = \widetilde{O}(n_i)$ and $\Delta = \widetilde{O}(\mu)$, therefore $\Delta = \widetilde{O}(n_i^{1/2} \cdot \mu^{1/2})$. As we saw above, $m_i = \widetilde{\Omega}(n^{4/3}_i)$, as so $n_i = o(\widetilde{m}_i)$. Putting this together we get that $\Delta = o(\widetilde{m}_i^{1/2}\cdot \mu^{1/2})$ and therefore $\Delta \leq (\mu\widetilde{m}_i)^{1/2}/(20\log(n_i)) = \widetilde{m}_i\cdot p/(20\log(n_i))$, as required. Further, $400 \log^2(n_i) \leq \mu$, as $\mu$ is larger than $\Delta$ and every node in $C$ has degree which is at least polynomial. On the other hand, $\mu = p^2\cdot \widetilde{m_i}$, and so $400 \log^2(n_i) \leq \widetilde{m_i}$, as required. As the conditions of the lemma hold, the number of edges between $B_j$ and $B_k$ is at most $6p^2\cdot \widetilde{m_i} = O(\mu)$, w.h.p.

Each node in $K_{j^*}$ is uniquely assigned at most $2k_v \log(n_i)  
\cdot (\widetilde{m}_i/\mu)^{3/2} / n_i$ triples of sets $B_i$. Notice that there are $s^3 = (\widetilde{m}_i/\mu)^{3/2}$ ways to choose triplets $B_j$, and so since $2k_v \log(n_i) \geq n_i$, then this is enough triplets per node in $K_{j^*}$ to ensure that every possible choice of triplets $B_j$ is chosen by a node in $K_{j^*}$.
For each node $v \in K_{j^*}$, denote the set of triplets of sets $B_j$ it is assigned using $T(v)$. 

Finally, for every $v \in K_{j^*}$, node $v$ iteratively learns triplets from $T(v)$, lists the triangles it sees using these triplets, throws away this information and continues on to the next triplets. Notice that each triplet requires learning $O(\mu)$ edges and as such this is accomplished in $\mu/(2^{j^*-2} \cdot m_i/n_i) \cdot n^{o(1)}_i$ rounds using \Cref{congest-general-routing1}. This is repeated at most $2k_v \log(n_i)  
\cdot (\widetilde{m}_i/\mu)^{3/2} / n_i$ times, for a total complexity of $(2k_v \log(n_i)  
\cdot (\widetilde{m}_i/\mu)^{3/2} / n_i) \cdot \mu/(2^{j^*-2} \cdot m_i/n_i) \cdot n^{o(1)}_i$ rounds. As $k_v \leq 2^{j^*-1}$, this comes out to $((\widetilde{m}_i/\mu)^{3/2} / n_i) \cdot \mu/(m_i/n_i) \cdot n^{o(1)}_i = ((\widetilde{m}_i/\mu)^{3/2} / n_i) \cdot \mu n_i/m_i \cdot n^{o(1)}_i = (\widetilde{m}_i/\mu)^{3/2} \cdot \mu/m_i \cdot n^{o(1)}_i = (\widetilde{m}_i/\mu)^{1/2} n^{o(1)}_i$ rounds. As $\widetilde{m_i}^{1/2} = \widetilde{O}(n_i)$, we get a final complexity of $n_i^{1+o(1)}/\mu^{1/2} = n^{1+o(1)}/\mu^{1/2}$, as required.

\section{Random Order Streams}\label{subsec:random-order}
It is often beneficial to allow random order streams, as it may allow better space-passes-approximation ratio tradeoffs than adversarial order ones (e.g.,~\cite{guha2009stream}).
In our context, we show that (assuming $\mu=\Omega(n+\Delta^2)$) one can still assume adversarial inputs, and use the network to randomly shuffle the data. 
As in the above, we start by storing all edges at the neighbors of the highest degree node $v$, which is picked as the simulator.
Without loss of generality, we assume that every neighbor has exactly $n$ {edges (as we can add ``dummy edges'' that will not affect our simulation).}

Before formally giving our \switch algorithm, we describe our approach. 
We would like $v$ to generate a random permutation over the $n\cdot \Delta$ edges and have its neighbors stream the edges according to the selected order.
Note that a na\"ive approach is to let $v$ explicitly generate a random permutation and store the shuffled edges in the neighbors so that in the next phase, we stream the edges back to $v$ and their order in the stream is random. However, such an approach requires $\log (n\cdot \Delta)! \approx n\cdot\Delta\cdot \log n$ bits at $v$, which is too costly with respect to memory. 

Instead, we show we can distributedly implement the Fischer-Yates random shuffling algorithm~\cite{fisher1953statistical} using only $O(\Delta^2)$ memory. The Fischer-Yates algorithm for shuffling an array of size $s$ works as follows. Let $A$ be the array we wish to shuffle and let $A'$ be the output (shuffled) array. Elements in $A$ can be ``struck out'' to denote that they have already been selected. Go over elements from 0 to $s-1$, the $i$-th element in $A'$ is selected uniformly at random from the non-struck out elements of $A$ and is then struck out. 

Let us consider the elements that we wish to shuffle as an array of size $n\cdot\Delta$. Consider the following equivalent formulation of the algorithm: Divide $A$ into $\Delta$ consecutive buckets of size $n$, that is, each bucket is with indices $B_k = [nk, n(k+1)-1]$ for $0\leq k\leq \Delta-1$. Let us denote by $a_k$ the number of elements in the $k$-th bucket that are not struck out.
Instead of picking $A'[i]$ directly, we first pick a bucket, where the $k$-th bucket is picked with probability $a_k / (n \Delta - i)$, and then we pick an element from the bucket uniformly at random. We strike out this selected element. Finally, we shuffle every non-overlapping consecutive $\Delta$ elements of $A'$ sequentially, using the original Fischer-Yates approach. The last step is redundant in the sequential algorithm, but will be critical for showing equivalence to our distributed implementation.

Formally speaking, we achieve the following result.

\StreamingRandom*

\begin{proof}
Our distributed shuffling algorithm has two stages. First, we randomly partition the edges between $v$'s neighbors. Then, the neighbors stream the edges one by one while $v$ is shuffling the $\Delta$ incoming edges at each round before inserting them into the random-order streaming algorithm. Intuitively, the last shuffle (corresponding to the last step of the ``bucketized'' Fischer-Yates algorithm described above) is done in order to remove any dependency on the order in which edges arrive within a round. 


In our \switch implementation of a random order stream at the simulator node $v$, the node $v$ locally decides on an arbitrary order of edges with respect to nodes, that is, it considers the first $n$ indices in the (unshuffled) array $A$ to be those at neighbor $v_0$, the next to be at neighbor $v_1$, and so on up to neighbor $v_{\Delta-1}$. The node $v$ executes the first step (bucket selection) of the equivalent formulation of the Fischer-Yates algorithm described above, and then delegates the second step (item selection within the bucket) to the neighboring nodes. Specifically, it is enough to provide every neighbor $v_{\ell}$ with a list of size $\Delta$ (denoted as $L_{\ell}$) that contains how many edges from its bucket need to be moved to every other neighbor $v_j$. $v_{\ell}$ then generates the relevant sets of edges for each other neighbor of $v$. That is, $v_{\ell}$ randomly partitions its edge-set such that $L_{\ell}[0]$ edges will be sent to $v_0$, $L_{\ell}[1]$ edges shall be sent to $v_1$, etc.
At the end of this step, for every edge $e$, neighbor $v_{\ell}$ knows the target neighbor that $e$ will be sent to in the next step. In the next step, edges are then routed to their target neighbor through the node $v$ and are shuffled by the target neighbors, which guarantees this is equivalent to the Fischer-Yates algorithm.

The way that the lists $L_{\ell}$ are created at $v$ is as follows. Initialize a $\Delta \times \Delta$ matrix $L$ to be all zeros. Iterate over $i$ from 0 to $n\cdot \Delta - 1$ and select a bucket index, and increment an entry in the matrix as follows. Using the above formulation, let $j$ be the index of $A$ that lands in $A'[i]$ in the permutation. Then, we increment $B[\floor{i / n} ][\floor{j / n}]$ to record that another edge should be passed from bucket $\floor{j / n}$ to bucket $\floor{i / n}$. Now, for $0\leq j \leq \Delta-1$, the $j$-th neighbor $v_{j}$ of $v$ receives the $j$-th column of the matrix $B$ as its list $L_j$. 

\paragraph{Avoiding congestion.}
Using the above approach indeed allows every neighbor of $v$ to compute a set of edges it must send to every other neighbor of $v$ such that the resulting order of edges within nodes is a random permutation. However, invoking the above na\"ively can lead to congestion, as it might be the case that many edges are sent to a single \mbox{node in one round (and perhaps none at other rounds).}

To overcome this, we observe that $(1/n)\cdot B$ (multiplying each coordinate of $B$ by $1/n$) gives a doubly stochastic matrix (whose sums over each row and each column are exactly 1). By Birkhoff's Theorem~\cite{birkhoff1946tres}, this implies that there exist $\Delta$ permutation matrices $P_0,\ldots,P_{\Delta-1}$ and coefficients $\gamma_0,\ldots,\gamma_{\Delta-1} > 0$ such that $(1/n)\cdot B = \sum_{i=0}^{\Delta-1} \gamma_i\cdot P_i$. Specifically, this means each such permutation matrix corresponds to a perfect matching in the bipartite graph $H=(U^1,U^2,E_H)$, where $U^1$ and $U^2$ are both copies of $[\Delta]$, and in which there exists an edge $(i,j)\in E_H$ iff $B[i][j]>0$. We use this matching to schedule one iteration without congestion, by having the simulator node $v$ invoke this approach to create the permutation matrices and thus compute a perfect matching (e.g., using alternating paths) in $O(\Delta^2)$ memory, and then it lets each neighbor know where to send an edge to. For each edge $(i,j)$ in the perfect matching, we decrement $B[i][j]$ to reflect that there is one fewer edge that needs to be routed from $i$ to $j$. This is essentially equivalent to setting $B=B-P$, where $P$ is the permutation matrix that corresponds to our perfect matching. As a consequence, $B$'s rows and columns now sum up to $n-1$, and $(1/(n-1))\cdot B$ is doubly stochastic.
This process can then be repeated $n$ times until all edges are partitioned as desired.
\end{proof}


\end{document}